\documentclass{article}
\usepackage[utf8]{inputenc}
\usepackage{amsmath,amssymb,mathtools,amsthm}
\usepackage[margin=1.0in]{geometry}
\usepackage{enumerate}
\usepackage{graphicx}
\usepackage{fancyhdr}
\usepackage[mathscr]{euscript}
\usepackage{mathrsfs}
\usepackage{caption}

\usepackage{listings}
\usepackage{xcolor}
\usepackage{braket}

\usepackage{hyperref}
\usepackage{cleveref}


\numberwithin{equation}{section}

\theoremstyle{plain}
\newtheorem{definition}{Definition}[section]
\newtheorem{theorem}[definition]{Theorem}
\newtheorem{remark}[definition]{Remark}
\newtheorem{corollary}[definition]{Corollary}
\newtheorem{proposition}[definition]{Proposition}
\newtheorem{lemma}[definition]{Lemma}

\newcommand{\dd}{\mathrm d}
\newcommand{\one}{\mathbf{1}}
\newcommand{\Ran}{\mathrm{Ran} \,}
\newcommand{\tr}{\mathrm{Tr}}
\newcommand{\hc}{\mathrm{h.c.}}
\newcommand{\supp}{\mathrm{supp}\,}
\newcommand{\Spec}{\mathrm{Spec}\,}
\newcommand{\s}{\mathrm{s}}

\begin{document}

	\author{S. Fournais, L. Junge, T. Girardot, L. Morin, M. Olivieri, A. Triay}
	\title{The free energy of dilute Bose gases at low temperatures\\ interacting via strong potentials}
	\date{\today}

	\maketitle

	\abstract{We consider a dilute Bose gas in the thermodynamic limit and prove a lower bound on the free energy for low temperatures which is in agreement with the conjecture of Lee-Huang-Yang on the excitation spectrum of the system. Combining techniques of \cite{FS2} and \cite{HHNST}, we give a simpler and shorter proof resolving the case of strong interactions, including the hard-core potential.}

%	\tableofcontents

%%%%%%%%%%%%%%%%%%%%%%%%%%%%%%%%%%%%%%%%%%%%%%%%%%
%%%%%%%%%%%%%%%%%%%%%%%%%%%%%%%%%%%%%%%%%%%%%%%%%%
\section{Introduction}
The mathematical analysis of the interacting Bose gas has been subject of intense scrutiny and important breakthroughs in the last decades. In particular, much effort has been dedicated to rigorously justifying the seminal work of Lee, Huang and Yang \cite{LHY} where they compute the spectral properties of the dilute Bose gas. A corner stone of this program has been the proof of the leading term of the ground state energy density in \cite{LY}, wherein they complement Dyson's upper bound \cite{dyson} from 40 years before with a matching lower bound. The successes of the approaches initiated around that time are well described in the lecture notes \cite{greenbook}. 

According to \cite{LHY} a dilute gas of Bosons behave similarly as a system of independent (quantum) harmonic oscillators following Bogoliubov's dispersion relation. Interestingly, the authors of \cite{LHY} gave their names to the correction to the leading order of the ground state energy density - the famous Lee-Huang-Yang term - which corresponds to the zero point energy of such a system, and was rigorously derived in \cite[upper bound]{YY} and \cite[lower bound]{FS,FS2}.

To see the effect of the rest of the spectrum, it is necessary to consider the system at positive temperature and, for example, compute its free energy per unit volume. This was recently investigated in the work \cite{HHNST} (see also \cite{haberberger2024upper} for the matching upper bound) where a two-term asymptotic expansion for the free energy per unit volume was proved for temperatures that are low but for which the thermal correction is of the order of the Lee-Huang-Yang (LHY) term.

Since the foundational work \cite{LHY}, a special role has been played by the hard core (or hard sphere) potential. This potential is often chosen for its simplicity as it depends on one parameter only (its scattering length, which equals its radius) and because the system in the dilute regime is not expected to depend on the microscopic details of the potential, at least to some order. However, its rigorous mathematical analysis is very challenging since it strongly penalises approximation schemes by having infinite integral. It is therefore an important testing problem for techniques in the area and for the idea of \textit{universality}, i.e. that the scattering length should be the main parameter governing the dilute regime.

The purpose of the current paper is two-fold. First of all we extend the result of \cite{HHNST} to the case of interaction potentials with large (possibly infinite) integral, thereby in particular for the first time rigorously obtaining the lower bound on the free energy of the Bose gas in the case of hard core interactions.
The second purpose is expository. By combining the Neumann localization approach of \cite{HHNST} (see also \cite{zbMATH07681333}) with the techniques of \cite{FS,FS2}, we provide a fairly simple and short proof of lower bounds in the thermodynamic limit in a generality allowing for very singular potentials. 
Our proof is clearly shorter than \cite{FS2} which is restricted to $T=0$ (although we here have slightly more restrictive assumptions on the potential) and is also arguably shorter and simpler than \cite{HHNST}.

We will now define the model considered and state the result of the paper.
We consider a system of $N$ interacting, non-relativistic bosons in a large box of sidelength $L$ in $3$ dimensions governed by the Hamiltonian
\begin{equation}\label{eq.hamiltonian}
	H_{N,L}=\sum_{i=1}^N-\Delta_i+\sum_{i<j}V(x_i-x_j).
\end{equation}
The operator $H_{N,L}$ acts on the bosonic Hilbert space $\otimes_{\s}^N L^2(\Lambda_L) =L^2_{\rm{sym}}(\Lambda^N)$, with $\Lambda_L := [0,L]^3$.
For concreteness, we will take Neumann boundary conditions in order to realize $H_{N,L}$ as a self-adjoint operator.
We define the ground state energy $E(L,N)$ and the free energy $F(L,N)$ at temperature $T>0$ by
\begin{equation}\label{eq. free energy}
	E(L,N) = \inf \Spec H_{N,L}, \qquad 
	F(L,N) = -T \log \tr \big(e^{-\frac{H_{N,L}}{T}} \big).
\end{equation}
We recall that the free energy is the minimum of the variational problem
\begin{equation}\label{eq.FLNvar}
F(L,N) = \inf_{\Gamma} \Big\lbrace {\rm{Tr}} \big( H_N \Gamma\big) - T\,S(\Gamma)   \Big\rbrace,
\end{equation}
where the infimum is taken over trace-class operators $\Gamma$ on $L^2_{\rm{sym}}(\Lambda^N)$ such that $0 \leq \Gamma \leq 1$ and ${\rm{Tr}}\, \Gamma =1$ and $S(\Gamma) = -\rm{Tr}(\Gamma \ln \Gamma)$ is the entropy of $\Gamma$.
The particle density for finite systems is $\rho = N/L^3$, and we define the ground state energy density $e(\rho)$ and the free energy density $f(\rho, T)$ by
\begin{align}\label{eq:def_f}
e(\rho) =  \lim_{\substack{N \to \infty \\ N/L^3 \to \rho}} \frac{E(L,N)}{L^3},\qquad
f(\rho,T) = \lim_{\substack{N \to \infty \\ N/L^3 \to \rho}} \frac{F(L,N)}{L^3}.
\end{align}
It is standard \cite{Ruelle} that the limits exist and are independent of both the boundary conditions and the sequence of $(N,L)$ as long as $N/L^3 \to \rho$.

The main result of the paper is the following.

\begin{theorem}\label{TL.thm}(Free energy in the thermodynamic limit)
For all $C_0>0$ there exists $C>0$ such that, for $\eta >0$ small enough the following holds. Let $V$ be a non-negative, radially symmetric and non-increasing potential with scattering length $a$ and compact support of radius $R \leq C_0 a$. Then for $0 \leq \rho a^3 \leq C^{-1}$, $\nu \in (0,\frac{\eta}{3})$ and $0\leq T \leq \rho a (\rho a^3)^{-\nu}$,
\begin{align}\label{eq:free_energy}
f(\rho,T) \geq 4\pi a \rho^2 \left(1 + \frac{128}{15 \sqrt \pi} \sqrt{\rho a^3}\right) + \frac{T^{5/2}}{(2\pi)^3} \int_{\mathbb{R}^{3}} \log\left(1-e^{-\sqrt{p^4 + \tfrac{16 \pi \rho a}{T} p^2}}\right)\dd p  - C(\rho a)^{5/2}(\rho a^3)^{\eta/4}.
\end{align}
\end{theorem}

\begin{remark}$\,$
\begin{itemize}
\item 
The result on the free energy in Theorem~\ref{TL.thm} confirms the picture of \cite{LHY}. According to this, the energy of a large dilute system should approximately be given by
\begin{align}\label{eq:LHY}
E(L,N) \approx 4\pi N \rho a \left(1 + \frac{128}{15 \sqrt \pi} \sqrt{\rho a^3}\right). 
\end{align}

Furthermore, the excitation spectrum for low-lying eigenvalues should be given by the Bogoliubov dispersion relation
$$
\sum n_p \sqrt{p^4 + 16 \pi \rho a p^2},
$$
where $p \in \frac{2\pi}{L} \mathbb Z^3 \setminus \lbrace 0 \rbrace$ denotes the momentum of an excitation, and $n_p\in {\mathbb N}$ denotes the number of bosonic excitations with that momentum. Combining these two pieces of information leads to the expression given in Theorem~\ref{TL.thm}.
Note that the expansion \eqref{eq:free_energy} is not expected to hold when the temperature reaches the transition temperature for Bose-Einstein condensation, which should be of order $\rho a (\rho a^3)^{-1/3}$. A leading order analysis in this temperature regime has been carried out in \cite{zbMATH05279032,yin}.
\item As mentioned above, the leading term in \eqref{eq:LHY} was proved by \cite{dyson,LY} in the thermodynamic limit. Further progress was \cite{ESY} (upper bound to second order for soft potentials) and \cite{zbMATH05594072,BSol} (two term asymptotics for very soft potentials).
The upper bound in the thermodynamic limit consistent with the two-term asymptotics \eqref{eq:LHY} was given by \cite{YY} (see also \cite{BCS,newupperbound}). The lower bound in the thermodynamic limit was proved in \cite{FS,FS2}.

There is extensive work on confined systems, in particular in the Gross-Pitaevskii scaling regime. We will not review this here but only refer to \cite{BCC} for an overview and further references (see also \cite{bocdeuchert,capdeuchert} for recent developments).

The energy of the dilute Bose gas is strongly dependent on the dimension of the ambient space, with the discussion above concerning only the physical $d=3$. Also the $1$ and $2$-dimensional situations are interesting and physically realizable. The $1$-dimensional gas with $\delta$ interactions was treated by \cite{zbMATH03224473} and more recently \cite{agerskov2024ground} treated the case of general interactions. In $2$ dimensions, the leading order energy was proved by \cite{zbMATH01638213}, and the correction term analogous to the Lee-Huang-Yang term was proved recently in \cite{2DLHY}.
The free energy in $2$-dimensions was analyzed to leading order in \cite{zbMATH07277873,zbMATH07201541}.
\end{itemize}

\end{remark}

The main part of our analysis relies on estimates on the free energy on boxes $\Lambda_\ell$ with 
\begin{align*}
\ell = K_\ell \frac{1}{\sqrt{\rho a}}, \qquad K_\ell =(\rho a^3)^{-\eta},
\end{align*}
where $\eta>0$ is the parameter in Theorem \ref{TL.thm}.
It is reasonably straightforward to prove Theorem~\ref{TL.thm} based on a result for the localized free energy given as \Cref{main.thm} below and the subadditivity of the free energy. The details of the proof of Theorem~\ref{TL.thm} based on \Cref{main.thm} are given in Appendix~\ref{sec:redtoboxes}. 

Before stating the result on the localised energy, let us recall that we defined our operator $H_{N,L}$ to have Neumann boundary conditions. At the scale $L=\ell$ the effect of boundary conditions is important, and the boundary conditions are also relevant for the proof of Theorem~\ref{TL.thm} based on \Cref{main.thm}.

\begin{theorem}(Free energy on $\Lambda_\ell$)\label{main.thm}
For all $C_0>0$, there exists $C>0$ such that, such that for $ 0<\eta<\frac{1}{1026}$ the following holds. Let $V$ be a positive, radially symmetric and non-increasing potential with scattering length $a$, and compact support of radius $R \leq C_0 a$. Then for all $\rho$ such that $0 < \rho a^3 \leq C^{-1}$, $0 \leq T \leq \rho a (\rho a^3)^{-\nu}$ with $\nu < \eta /3$ and $n \leq 20 \rho \ell^3$ we have
\begin{equation}\label{eq:main.thm}
F(\ell,n) \geq  F_{\mathrm{Bog}}(\ell,n) - C \ell^3 (\rho a)^{5/2}(\rho a^3)^{\eta/4},
\end{equation}
where
\begin{align} \label{eq.Fbogln}
F_{\mathrm{Bog}}(\ell,n) = 4\pi \rho_{n,\ell}^2 a\ell^3\Big( 1+ \frac{128}{15 \sqrt \pi} \sqrt{ \rho_{n,\ell} a^3} \Big) + T \sum_{p \in \frac{\pi}{\ell} \mathbf N_0^3} \log \Big( 1- e^{- \frac{1}{T}\sqrt{p^4 + 16\pi a \rho_{n,\ell} p^2}}\Big),
\end{align}
with $\rho_{n,\ell} = n \ell^{-3}$.
\end{theorem}
%%%%%%%%%%%%%%%%%%%%%%%%%%%%%%%%%%%%%%%%%%%%%%%%%%
%%%%%%%%%%%%%%%%%%%%%%%%%%%%%%%%%%%%%%%%%%%%%%%%%%

In the following Section~\ref{sec:proof} we streamline  the proof of Theorem~\ref{main.thm} by shifting the proofs of the main technical lemmas to the next sections.
Since most of the paper will be focused on the box of size $\ell$ we will often simplify notation and write $\Lambda = \Lambda_{\ell}$.

\section{Free energy on \texorpdfstring{$\Lambda_\ell$}{Lambda}}\label{sec:proof}

This section is devoted to the proof of Theorem \ref{main.thm} (given at the end of this section) which is the main ingredient in the proof of our main result Theorem \ref{TL.thm}. 
\subsection{Replacement by an integrable potential}\label{sec:replacement_pot}
We begin by replacing the initial interaction potential $V$ by a suitable integrable potential $v\leq V$, which obviously lowers the free energy. We need however to control the difference between the scattering lengths to ensure that the error coming from the substitution in the leading term does not affect the correction term.

Let us first recall some facts about the scattering problem. We denote by $\varphi_V$ the radial solution in ${\mathbb R}^3$ to the zero-energy scattering equation
\begin{equation}\label{eq:scatteq}
 - \Delta \varphi_V + \frac 1 2 V \varphi_V = 0,
 \end{equation}
normalized such that $\lim_{|x| \to \infty} \varphi_V(x) = 1$. Of particular interest throughout the paper are also the functions  
\begin{equation}
\omega_V := 1-\varphi_V,\qquad
g_V := V \varphi_V = V(1-\omega_V).
\end{equation}
By these functions, it is possible to rewrite equation \eqref{eq:scatteq} as 
\begin{equation}
-\Delta \omega_V = \frac{1}{2} g_V,
\end{equation}
and from this, by the divergence theorem, recover the scattering length by
\begin{equation}
 \widehat{g}_V(0) = \int g_V(x) \dd x = 8 \pi a.
\end{equation}
We refer to Section \ref{sec.scattering} for more details.

\begin{proposition}[Replacement by an integrable potential]\label{prop.L1v}
There exists $C>0$ such that the following holds. Let $V$ be a positive, radially symmetric and decreasing potential with scattering length $a$, and compact support of radius $0 \leq R \leq C_0a$. Then we can construct another positive and radially symmetric potential $v \leq  V$ with scattering length $a(v)$ satisfying
\begin{equation} \label{eq.propv1}
0 \leq a - a(v) \leq a \sqrt{\rho a^3} K_\ell^{-1}, \qquad \int v(x) \dd x \leq C (\rho a)^{-\frac 1 2} K_\ell, \qquad  v \leq \ell^2 a^{-4},
\end{equation}
and
\begin{equation} \label{eq.propv2}
g_v(y) \leq C  v(x) \quad \text{ for } \quad |x| \leq |y| ,
\end{equation}
where $g_v= v  \varphi_v$ and $\varphi_v$ is the scattering solution for $ v$.
\end{proposition}
The proof is given in Section \ref{sec.scattering}. We denote here by $F(\ell,n)[V]$ and $F_{\mathrm{Bog}}(\ell,n)[a]$ the free energy \eqref{eq.FLNvar} associated to the potential $V$ in the box $\Lambda$ and the expression \eqref{eq.Fbogln} with scattering length $a$, respectively. Assuming that the desired estimate (\ref{eq:main.thm}) holds for the potential $v$ and using that the free energy $F(\ell,n)[V]$ is a non-decreasing function of $V$, we obtain
\begin{align*}
F(\ell,n)[V] 
	\geq F(\ell,n)[v] 
	&\geq F_{\mathrm{Bog}}(\ell,n)[a(v)] - C \ell^3 (\rho a(v))^{5/2} (\rho a(v)^3)^{\eta/4} \\
	&\geq F_{\mathrm{Bog}}(\ell,n)[a] - C \ell^3 (\rho a)^{5/2} (\rho a^3)^{\eta/4} - C\left(\ell^3 \rho^2 + T^{3/2} \rho\ell^3\right) |a(v) - a| \\
	&\geq F_{\mathrm{Bog}}(\ell,n)[a] - C \ell^3 (\rho a)^{5/2}  (\rho a^3)^{\eta/4},
\end{align*}
where we used (\ref{eq.propv1}) to obtain the second line and \cite[Lemma 9.1]{HHNST} for the replacement of $a(v)$ by $a$ in the sum appearing in $F_{\mathrm{Bog}}(\ell,n)$.

Therefore, it is enough to prove Theorem \ref{main.thm} for potentials $v$ satisfying  \eqref{eq.propv1} and \eqref{eq.propv2}. In the sequel we will abuse notation and write $F(\ell,n)$ for $F(\ell,n)[v]$, and also write $g,\omega$ and $\varphi$ for $g_v,\omega_v$ and $\varphi_v$.

\subsection{Splitting of the potential energy} 
Let us introduce the projections on and outside the condensate,
\begin{equation}\label{def:split}
	P= \frac{1}{\ell^3} \vert  1 \rangle\langle 1 \vert ,\qquad Q=\one-P,
\end{equation}
acting on $L^2(\Lambda)$. If $\Phi \in L^2(\Lambda^N)$ and $1 \leq i \leq N$, we denote by $P_i \Phi$ and $Q_i\Phi$ the action of $P$ and $Q$ on the variable $x_i$. We also denote the number of particles in the condensate and the number of excited particles respectively by
\begin{equation}
n_0 = \sum_{i=1}^N P_i, \qquad n_+ = \sum_{i=1}^N Q_i = N - n_0.
\end{equation}
The strategy carried out in \cite{LewNamSerSol-15} and used in many other works, is to expand $\one_{(L^2)^{\otimes_{\textrm{s}} N}} = \prod_{i=1}^N (P_i+Q_i)$ which leads to identifying $(L^2)^{\otimes_{\textrm{s}}  N}$ with a truncated Fock space over the excitation space $\Ran Q$ where the vacuum vector is the pure condensate $\ket{\frac{1}{\ell^{3/2}}}^{\otimes_{\textrm{s}} N}$. Having factored out the condensate, it remains to extract the correlation energy. This strategy does not work for hard core potentials because the vacuum (the pure condensate) does not even belong to the domain of the Hamiltonian: one needs to factor out parts of the correlations first. In this spirit, the strategy initiated in \cite{FS} and then \cite{FS2, 2DLHY} is to \textit{renormalize} the potential energy from the beginning by using the identity 
\begin{align}\label{eq:corr_decompo}
\one= (Q_iQ_j + \omega(x_i-x_j)(\one-Q_iQ_j)) + \varphi(x_i-x_j)(\one-Q_iQ_j).
\end{align}

\begin{lemma}\label{lem:split}
	For all positive and radially symmetric potentials $v$, we have the identity
	\begin{equation*}
		\sum_{i<j}v(x_i-x_j)={Q}_0^{\rm ren}+ Q_{1}^{\rm ren}+ Q_{2}^{\rm ren}+ Q_{3}^{\rm ren}+ Q_4^{\rm ren}
	\end{equation*}
	where
	\begin{align}
		{Q}_0^{\rm ren}&:= \frac{1}{2} \sum_{i \neq j} P_i P_j (g+ g\omega)(x_i-x_j) P_j P_i. \label{eq:SF_DefQ0} \\
		{Q}_1^{\rm ren}&:=\sum_{i \neq j} \big(Q_i P_j (g + g \omega)(x_i-x_j) P_j P_i  + \hc\big ),  \label{eq:SF_DefQ1}\\
		{ Q}_2^{\rm ren}&:=
		\sum_{i\neq j} P_i Q_j (g+ g\omega)(x_i-x_j)P_j Q_i   + \sum_{i\neq j} P_i Q_j (g+g\omega)(x_i-x_j)Q_j P_i   \nonumber\\
		&\quad+\frac{1}{2}\sum_{i\neq j} P_iP_j g(x_i-x_j) Q_j Q_i + \hc,
		\label{eq:SF_DefQ2}\\
		{Q}_3^{\rm ren}&:=
		\sum_{i\neq j} P_i Q_j g(x_i-x_j) Q_j Q_i + \hc,
		\label{eq:SF_DefQ3}\\
		{Q}_4^{\rm ren}&:=
		\frac{1}{2} \sum_{i\neq j} \Pi_{ij}^* v(x_i-x_j) \Pi_{ij}\, , \label{eq:SF_DefQ4}
	\end{align}
	with
	\begin{align*}
	 \Pi_{ij} := Q_j Q_i + \omega(x_i-x_j) \left(P_j P_i + P_j Q_i + Q_j P_i\right).
	\end{align*}
\end{lemma} 

\begin{proof}
This is proved by inserting (\ref{eq:corr_decompo}) in each term of the potential $\sum_{i<j}v(x_i-x_j)$ and expanding using that $1-Q_iQ_j = P_iP_j + P_iQ_j + Q_iP_j$. We recall that $g = v(1-\omega)$.
\end{proof}

%%%%%%%%%%%%%%%%%%%%%%%%%%%%%%%%%%%%%%%%%%%%%%%%%%
%%%%%%%%%%%%%%%%%%%%%%%%%%%%%%%%%%%%%%%%%%%%%%%%%%

\subsection{Spectral gaps} 
Many of our estimates will rely on analysis in momentum space. In this article, the momenta are elements of $\frac{\pi}{\ell} \mathbb N_0^3$, and we define low and high momenta by
\begin{equation}\label{def:PLPH}
\mathcal P_L := \Big\lbrace p \in \frac{\pi}{\ell} \mathbb N_0^3 , \;\; 0< \vert p \vert \leq K_H \ell^{-1} \Big\rbrace, \qquad \mathcal P_H := \Big\lbrace p \in \frac{\pi}{\ell} \mathbb N_0^3, \;\; \vert p \vert > K_H \ell^{-1} \Big\rbrace,
\end{equation}
where $K_H$ is a large $\rho$-dependent parameter. 
In the theorems below we always write our assumptions on $K_H$. The final choice of $K_H$, and a few other parameters, will be specified in \eqref{eq:params} below.
We define the corresponding localized projectors by
\begin{align}\label{def:QL}
Q^L := \one_{\mathcal P_L} (\sqrt{-\Delta}), \qquad Q^H := \one_{\mathcal P_H} (\sqrt{-\Delta}),
\end{align}
where $\Delta$ is the Neumann Laplacian on $\Lambda$. The number of high and low excitations are respectively given by
\begin{equation}\label{def:n+H}
n_+^H := \sum_{j=1}^N  Q_{j}^H, \qquad n_+^L := \sum_{j=1}^N Q^L_{j},
\end{equation}
both acting on $L^2_{\rm{sym}}(\Lambda^N)$. These definitions imply for instance that $P+Q_L +Q_H = 1$ and that $n_+^L+n_+^H = n_+$.

Many of our error terms will be bounded by $n_+$, $n_+^H$ or $n_+^L$. To control these errors, we will extract some positive quantities from the kinetic energy, refered to as \textit{spectral gaps}. They are gathered into one operator
\begin{equation}\label{eq:gap}
G :=\frac{\pi n_+}{4\ell^2}+\frac{K_H n^H_+}{2\ell^2}+\frac{ \pi n_+^Ln_+}{4 \mathcal{M}\ell^2}+\frac{K_H n_+^Ln^H_+}{2 \mathcal{M}\ell^2},
\end{equation}
where $\mathcal M$ is some large $\rho$-dependent parameter, which will need to satisfy specific conditions in later theorems. These spectral gaps are extracted from the kinetic energy, from which will only remain
\begin{equation*}
\sum_{j=1}^N   \mathcal T_j = \sum_{j=1}^N -\Delta_j -\frac{\pi n_+}{2\ell^2}-\frac{K_H n^H_+}{\ell^2},
\end{equation*}
with
\begin{equation}\label{eq:T_i}
\mathcal T=-\Delta -\frac{\pi }{2\ell^2}Q-\frac{K_H }{\ell^2}Q^H\geq 0.
\end{equation}
More precisely, we obtain the following result.

\begin{theorem}\label{thm.gaps}
There exist $C>0$ such that the following holds. Let $v$ be as in Proposition \ref{prop.L1v}, $\rho \ell^{3} (\rho a^3)^{\alpha} \leq N \leq 20 \rho \ell^3$, $T \leq \rho a (\rho a^3)^{-\nu}$, with $\alpha + 5\nu/2 < 6/17$ and  $\mathcal{M} \geq \rho \ell^3 (\rho a^3)^{\gamma}$, for some $\gamma >0$, then for $\rho a^3 \leq C^{-1}$ we have
\begin{align*}
F(\ell,N) &\geq \inf_{\Gamma} \Big\lbrace \tr \big( H_N^{\rm{mod}} \Gamma  + T \, \Gamma \ln \Gamma +  G \Gamma \big) \Big\rbrace - C\ell^3 (\rho a)^{5/2} (\rho a^3)^{1/18-2\gamma - \nu} K_H^{3},
\end{align*}
where the infimum is taken over trace-one operators $\Gamma \geq 0$. Here the modified Hamiltonian is
\begin{equation}\label{eq.Hmod}
H_N^{\rm{mod}} = \sum_{i=1}^{N} \mathcal T_i + \sum_{i<j}v(x_i-x_j),
\end{equation}
and $G$ was defined in \eqref{eq:gap}.
\end{theorem}

The proof of Theorem \ref{thm.gaps} is given in Section \ref{sec.gaps} and relies on rough a priori bounds on the numbers of excitations in $\Lambda_\ell$.
%
%%%%%%%%%%%%%%%%%%%%%%%%%%%%%%%%%%%%%%%%%%%%%%%%%%
%%%%%%%%%%%%%%%%%%%%%%%%%%%%%%%%%%%%%%%%%%%%%%%%%%

\subsection{Localization of the 3Q term}
Here we show that the main contribution of  $Q_3^{\rm{ren}}$ defined in (\ref{eq:SF_DefQ3}) comes from the creation and annihilation of the so-called ``soft pairs" where one particle is in the condensate and the other has low but non-zero momentum. 
\begin{lemma}\label{lem.Q3loc1}
Let $v$ be as in Proposition \ref{prop.L1v}. For all $\varepsilon >0$, there exists a $C>0$ such that the following holds. Assume that $\rho a^3 \leq C^{-1}$, then for all $K_H \geq C K_\ell^4$,
\[ Q_3^{\rm{ren}} \geq Q_{3,L} - \varepsilon \left( Q_4^{\rm{ren}} + \frac{1}{\ell^2} n_+ + \frac{ K_H}{\ell^2} n_+^H\right),\]
with
\begin{equation}\label{eq:defQ3low}
 Q_{3,L} := \sum_{i \neq j} P_i Q_{j}^L g(x_i - x_j)Q_i Q_j + \hc
\end{equation}
\end{lemma}

The last two error terms will be absorbed by a fraction of the spectral gaps $G$ by taking $\varepsilon$ sufficiently small but fixed. The proof of Lemma \ref{lem.Q3loc1} is given in Section \ref{sec:3Q}.

%%%%%%%%%%%%%%%%%%%%%%%%%%%%%%%%%%%%%%%%%%%%%%%%%%
%%%%%%%%%%%%%%%%%%%%%%%%%%%%%%%%%%%%%%%%%%%%%%%%%%

\subsection{Symmetrization}
To deal with Neumann boundary conditions, an important idea of \cite{HHNST} is to symmetrize the potential using a mirroring technique to make it commute with the Neumann momentum. The mirror transformations $p_z$, for $z\in \mathbb{Z}^3$, are defined by
\begin{equation}\label{def:pz}
(p_z(x))_i=(-1)^{z_i}\Big(x_i-\frac{\ell}{2}\Big)+ \frac{\ell}{2}+\ell z_i, \qquad i=1,2,3.
\end{equation}
Note that in particular $p_z(\Lambda) = \{x+\ell z\;:\;x\in \Lambda\}$.
For $p \in \frac{\pi}{\ell} \mathbb N^3_0$, we denote by
\begin{equation}\label{def:neubasis}
u_p(x)= \frac{1}{\sqrt{|\Lambda|}} \prod_{i=1}^3 c_{p_i} \cos (p_i x_i), \qquad \text{where } \quad c_{p_i} = \begin{cases}
1 &\text{if} \quad p_i = 0, \\
\sqrt{2} &\text{if} \quad p_i \neq 0,
\end{cases}
\end{equation}
the normalized eigenbasis of the Neumann Laplacian on $\Lambda$. For any $f \in L^1(\Lambda)$, the symmetrization of $f$ is defined by
\begin{equation}\label{def:gs}
f^{\rm{s}}(x,y) := \sum_{z \in \mathbb Z^3} f(p_z(x) - y)
\end{equation}
for a.e. $x,y \in \Lambda$. If $f$ is radial, as shown in \cite[Lemma 3.2]{HHNST}, the operator with kernel $f^{\rm{s}}(x,y)$ commutes with the Neumann basis.
\begin{lemma}\label{lem:sq}
Let $f:\; \mathbb{R}^3\to \mathbb{R}$ be radial and integrable with $\mathrm{supp}(f) \subset B(0,R)$ for some $R \leq \ell /2$. Then for $p,q \in \frac{\pi}{\ell} \mathbb{N}^3$ we have
\begin{equation}\label{eq:2u}
\int_{\Lambda^2}f^\s(x,y)u_p(x)u_q(y)\mathrm{d}x\mathrm{d}y= \delta_{p,q} \int_{\mathbb{R}^3} f(x)\prod_{i=1}^3 \cos (p_i x_i)\mathrm{d}x = \delta_{p,q}\widehat{f}(p).
\end{equation}
\end{lemma}
\noindent Here and through all the paper $\widehat{f}$ is the Fourier transform of $f$ defined by
\begin{equation*}
\widehat{f}(p)=\int_{\mathbb{R}^3} f(x)e^{-ipx}\mathrm{d}x.
\end{equation*}

\noindent Note that since $\supp g \subset \supp v \subset B(0,R)$, the sum defining $g^s$ is finite, and $g^s$ agrees with $g$ except when $x$ or $y$ is at distance $R$ from the boundary of $\Lambda$. We then define the symmetrized operators
\begin{align}\label{eq:Q_0sym}
Q_0^{\rm{sym}} &:= \frac 12 \sum_{i \neq j} P_i P_{j} (g + g \omega)^\s(x_i,x_j)P_i P_j,  \\
{Q}_1^{\rm sym}&:=\sum_{i \neq j} Q_i P_j (g + g \omega)^\s(x_i-x_j) P_j P_i  + \hc, \\
{Q}_2^{\rm sym}&:=
		\sum_{i\neq j} P_i Q_j (g+ g\omega)^\s(x_i-x_j)P_j Q_i   + \sum_{i\neq j} P_i Q_j (g+g\omega)^\s(x_i-x_j)Q_j P_i \\
		&\quad+\frac{1}{2}\sum_{i\neq j} P_iP_j g^s(x_i-x_j) Q_j Q_i + \hc, \nonumber\\
Q_{3,L}^{\rm{sym}} &:=  \sum_{i \neq j} P_i Q_{j}^L g^{\rm{s}}(x_i,x_j)Q_i Q_j + \hc \label{eq:Q_3sym}
\end{align}

\begin{theorem}\label{thm.sym}
For all $\varepsilon >0$, there exists a $C>0$ such that the following holds. Let $v$ be as in Proposition \ref{prop.L1v} and assume that $\rho a^3 \leq C^{-1}$, $K_H \geq C K_\ell^4$, and $K_\ell K_H^3 \leq (\rho a^3)^{- \frac 1 2}$, we have
\begin{align}
H_N^{\rm{mod} } &\geq H_N^{\rm{sym}} -  C N \rho a \frac R \ell - \varepsilon G, \nonumber
\intertext{where}
\label{def:hsym}
H_N^{\rm{sym}} &:= \sum_{i=1}^N \mathcal T_i  + Q_0^{\rm{sym}} + Q_{2}^{\rm{sym}}+  Q_{3,L}^{\rm{sym}}.
\end{align}
\end{theorem}

\noindent The proof of Theorem \ref{thm.sym} is given in Section \ref{sec.sym}.
\begin{remark}
The assumption that $V$ is non-decreasing is only used here. It ensures that $v$ satisfies the condition \eqref{eq.propv2} which is needed to estimate some error terms arising from the symmetrization of the term ${Q}_2^{\mathrm{ren}}$.
\end{remark}
%%%%%%%%%%%%%%%%%%%%%%%%%%%%%%%%%%%%%%%%%%%%%%%%%%
%%%%%%%%%%%%%%%%%%%%%%%%%%%%%%%%%%%%%%%%%%%%%%%%%%

\subsection{C-number substitution}
From now on, we work with the symmetrized Hamiltonian $H_N^{\rm{sym}}$ in second quantization.
Recall the basis $\{u_p\}_p$ defined in (\ref{def:neubasis}). We introduce the bosonic Fock space 
\begin{equation}
{\mathscr F}(L^{2}(\Lambda)) = \bigoplus_{n=0}^{\infty} L^2(\Lambda)^{\otimes_s n},
\end{equation}
and the associated creation and annihilation operators that satisfy the canonical commutation relations
\begin{align}
\label{eq:ccr}
a_p = a(u_p), \qquad a^*_p = a^*(u_p), \qquad \big[ a_p, a_q^* \big] = \delta_{p,q}
\end{align}
for all $p,q \in {\Lambda}^*=\frac{\pi}{\ell} \mathbb N^3_0$.  
It is convenient to extend this definition to $p =(p_1,p_2,p_3) \in \frac{\pi}{\ell} \mathbb Z^3$, by denoting $u_p = u_{(\vert p_1\vert, \vert p_2\vert,\vert p_3\vert)}$ and $a_p= a(u_p) = a_{(\vert p_1\vert, \vert p_2\vert,\vert p_3\vert)}$. We also define $\Lambda^*_+=\frac{\pi}{\ell} \mathbb N^3_0 \setminus  \lbrace 0 \rbrace$ and the set of generalized low momenta as
\[ \mathcal P_L^{\mathbb Z} = \big\lbrace p \in \frac{\pi}{\ell} \mathbb Z^3, \quad 0 < |p| \leq K_H \ell^{-1} \big\rbrace.\]
With this notation we obtain
\begin{lemma}[Second quantization in the Neumann basis]\label{lem.secondquant}
We have the following identity on the sector of $\mathscr{F}(L^2(\Lambda))$ with $N$ bosons:
\begin{align}
H_N^{\rm{sym}} &= \frac{\widehat{g}(0)(N(N-1)-n_+(n_+-1))}{2|\Lambda |} +\frac{\widehat{g\omega}(0)}{2|\Lambda |}   a_0^*a_0^* a_0a_0 \nonumber \\
	&\quad  + \sum_{p \in {\Lambda}^*_+} \Big( \tau(p) a_p^* a_p + \frac{\widehat{g}(p)}{|\Lambda |} a_0^* a_p^* a_p a_0 + \frac{ \widehat{g}(p) }{2 |\Lambda |}\big( a_0^* a_0^* a_p a_p + \hc \big) \Big) \nonumber \\
& \quad + \frac{1}{|\Lambda|} \sum_{\substack{k \in {\Lambda}^*_+,\, p \in \mathcal P_L^{\mathbb Z} \\p\neq k} }c(p,k) \widehat{g}(k) \big( a_0^* a_p^* a_{p-k} a_k + \hc \big) + \frac{1}{|\Lambda|} \sum_{p\in {\Lambda}^*_+}\Big((\widehat{g \omega}(0)+\widehat{g \omega}(p))a_0^* a_p^*a_p a_0 \Big). \label{eq:lem.secondquant}
\end{align}
where 
\begin{equation}\label{eq:tau}
\tau(p) = |p|^2 - \frac{\pi}{2\ell^2}\one_{\lbrace p \neq 0 \rbrace} - \frac{K_H}{\ell^2} \one_{\lbrace p \in \mathcal {P}_H \rbrace }
\end{equation}
is the symbol of the kinetic energy $\mathcal T$ and, recalling \eqref{def:neubasis}, $c(p,k)$ are the normalizing factors given by 
\begin{equation}
c(p,k):= \prod_{i=1}^3 \frac{c_{k_i-p_i}}{c_{p_i} c_{k_i}}.
\end{equation}
Note that if none of the indices $i$ in the product is $0$, then the above constant is $\frac{1}{\sqrt{8}}$.
\end{lemma}
The proof of Lemma \ref{lem.secondquant} is given in Section \ref{sec:cnum}. Following Bogoliubov's approximation and using \cite{LieSeiYng-05}, we perform a $c$-number substitution, effectively replacing $a_0$ by a complex number $z$ and $a_0^*$ by $\overline{z}$. The transformed Hamiltonian acts on the excitation Fock space,
\[ \mathscr F^\perp = \mathscr  F \big( \lbrace 1 \rbrace^\perp \big).\]
We obtain the following lower bound on the free energy of $H^{\rm{sym}}_N$.

\begin{theorem}\label{thm.cnumber}
For any $m >3$
there exist $C>0$ and $\varepsilon>0$ such that the following holds. Let $v$ be as in Proposition~\ref{prop.L1v} and assume $\rho a^3 \leq C^{-1}$. Then for all $0 \leq 10\mu \leq \ell^{-2}$, $C \leq N\leq 20 \rho \ell^3$, $0 \leq T \leq C \rho a (\rho a^3)^{-\nu}$, $C \leq \mathcal M \leq C^{-1} \ell/a$, and $K_H \geq C K_\ell^4$, $K_\ell K_H^3 \leq (\rho a^3)^{- \frac 1 2}$, we have
\[ - T\log \tr \Big(  e^{- \frac 1 T  (H_N ^{\rm{sym}} + \frac 1 2 G)} \Big) \geq 4\pi a \frac{ N^2 }{|\Lambda|} - T \log \int_{\mathbb C} \tr_{\mathscr F ^\perp} \Big(  e^{- \frac 1 T ( \mathcal H_\mu(z) + \varepsilon \mathcal G(z))} \Big) e^{- \frac{\mu}{T}  N} \dd z - C \rho a ,\]
where 
\begin{align}\label{def.Hmu}
		\mathcal{H}_{\mu}(z)&=  \sum_{p\in {\Lambda}^*_+}\Big(\tau(p) a_p^* a_p +\frac{\vert z\vert^2}{|\Lambda|}  \widehat g(p)   a_p^* a_p + \frac{1}{2|\Lambda |} \widehat g(p) ( \overline{z}^2 a_pa_{p} + z^2 a_p^* a_p^*)  \Big) \nonumber \\
		&\quad -\mu \vert z\vert^2 +8 \pi a(\rho a^3)^{\frac{1}{4}}\Big(\frac{\vert z\vert^4}{\vert \Lambda\vert }-\rho N\Big) \nonumber \\
		&\quad+ \frac{1}{2 \vert \Lambda \vert}  \widehat{g \omega}(0) \vert z\vert^4+\frac{\vert z\vert^2}{\vert \Lambda \vert }\sum_{p\in {\Lambda}^*_+}(\widehat{g \omega}(0)+\widehat{g \omega}(p))a_p^*a_p+ Q_{3,L}^{\rm{sym}}(z), \\
Q_{3,L}^{\rm{sym}}(z) &= \frac{1}{|\Lambda|} \sum_{\substack{k \in {\Lambda}^*_+, p \in \mathcal P_L^{\mathbb Z}\\ p\neq k} }c(p,k) \widehat{g}(k) \big( \overline{z} a_p^* a_{p-k} a_k + \hc \big), \label{def:Q3Lz}
	\end{align}
and
\begin{align}\label{def.Gz}
\mathcal G(z) = \frac{\pi^2 n_+}{2\ell^2} +  \frac{K_H n_+^H}{\ell^2} + \frac{n_+^L n_+}{\mathcal M \ell^2} + \frac{K_H n_+^L n_+^H}{\mathcal M \ell^2} + \frac{\rho a}{N^m}(|z|^{2m} + |z|^{2m-2} n_+ + |z|^{2m-4} n_+^2).
\end{align}
\end{theorem}

\noindent The proof of Theorem \ref{thm.cnumber}, including details of the c-number substitution, is given in Section \ref{sec:cnum}.
\begin{remark}\label{rem:convexity}$\,$
\begin{enumerate}
\item 
In the physically relevant region $|z|^2 \simeq N$ we control the number of particles using the chemical potential $\mu$. To estimate the contribution of the physically non-relevant region $|z|^2 \gg N$, we will use the last terms in \eqref{def.Gz}.
\item The term proportional to $(\rho a^3)^{\frac{1}{4}}$ in \eqref{def.Hmu} is artificially added to the energy to provide convexity in the variable $\vert z\vert^2$ (see \eqref{eq:Fz}).
\end{enumerate}
\end{remark}
For $|z|^2 \geq K_{\ell}^{1/4}N =  (\rho a^3)^{-\eta/4} N$ we have a simple lower bound on the Hamiltonian.
\begin{lemma}\label{lem.largez}
Under the assumptions of Theorem \ref{thm.cnumber}, if $|z|^2 \geq K_{\ell}^{1/4} N$ and $m > 2\eta^{-1} + 14$ such that $K_{H} \leq K_{\ell}^{\frac{m+1}{12}}$, then there exists a constant $c>0$ such that
\[ \mathcal H_\mu(z) + \varepsilon \mathcal G(z) \geq \frac 1 2 \sum_{p \in {\Lambda}^*_+} \tau(p) a_p^* a_p  +  c\rho a K_{\ell}^{\frac{m-1}{4}}\frac{|z|^{2}}{N},\]
where we recall that $\tau(p)$ is given by (\ref{eq:tau}).
\end{lemma}

\noindent The proof is given at the end of Section \ref{sec:cnum}. It remains to deal with $|z|^2 \leq K_{\ell}^{1/4}N$.

%%%%%%%%%%%%%%%%%%%%%%%%%%%%%%%%%%%%%%%%%%%%%%%%%%
%%%%%%%%%%%%%%%%%%%%%%%%%%%%%%%%%%%%%%%%%%%%%%%%%%

\subsection{Bogoliubov diagonalization}

We now diagonalize the main quadratic part of the Hamiltonian $\mathcal H_\mu(z)$ appearing in the first line of (\ref{def.Hmu}). The second line will be estimated later. For $|z|^2 \leq K_{\ell}^{1/4}N$, using the CCR (\ref{eq:ccr}), we obtain the identity
\begin{align}
	&\sum_{p\in {\Lambda}^*_+} (\tau(p) + \rho_z  \widehat g(p)) a_p^* a_p+ \frac{1}{2\vert \Lambda\vert} \widehat g(p) ( \overline{z}^2 a_pa_{p} + z^2 a_p^*a_p^*) \nonumber \\
	& = \sum_{p \in {\Lambda}^*_+} D_p(z) b_p^* b_p   + \frac{1}{2}\sum_{p\in {\Lambda}^*_+} \Big( \sqrt{\tau(p)^2 + 2 \rho_z \widehat g(p) \tau(p)} - \tau(p) - \rho_z \widehat g(p)  \Big), \label{eq.diag}
\end{align}
where we denoted $\rho_z = \frac{|z|^2}{ |\Lambda|}$ and
\begin{equation}\label{def:bdiagonalization}
D_p(z):=\sqrt{\tau(p)^2+2\tau(p)\rho_z\widehat{g}(p)},\qquad b_p:=\frac{a_p+\alpha_pa_p^*}{\sqrt{1-\alpha_p^2}},\qquad \alpha_p:=\frac{\tau(p)+\rho_z \widehat{g}(p)-D_p(z)}{\rho_z\widehat{g}(p)}.
\end{equation}
One easily checks that the argument of the square root is non-negative using that $|\widehat{g}(p) - \widehat{g}(0)| \leq R^2 \widehat{g}(0)|p|^2\leq C a^3 p^2$.

%%%%%%%%%%%%%%%%%%%%%%%%%%%%%%%%%%%%%%%%%%%%%%%%%%
%%%%%%%%%%%%%%%%%%%%%%%%%%%%%%%%%%%%%%%%%%%%%%%%%%

\subsection{Contribution of the 3Q terms.}

It remains to bound the terms in the second line of \eqref{def.Hmu}, this is done in the following lemma.

\begin{theorem}\label{thm.Q3}
Under the assumptions of Theorem \ref{thm.cnumber} and if $K_\ell^{5/4} K_H^2 \leq C^{-1} (\rho a^3)^{- \frac{1}{2}}$ and $K_H \geq K_\ell^4$, then for all $|z|^2 \leq K_{\ell}^{1/4}N$ and all $\mathcal M \leq C^{-1} \rho \ell^3 K_H^{-3} K_\ell^{-17/4}$ we have
\[ (1-K_H^{-1}) \sum_{k\in \mathcal P_H} D_k(z) b_k^* b_k + \frac{\vert z\vert^2}{\vert \Lambda \vert }\sum_{p\in {\Lambda}^*_+}\Big((\widehat{g \omega}(0)+\widehat{g \omega}(p))a_p^*a_p\Big) +  Q_{3,L}^{\rm{sym}}(z) \geq -\varepsilon \mathcal G(z) - C N \rho a \sqrt{\rho a^3} K_{\ell}^{-1}. \]
\end{theorem}

Theorem \ref{thm.Q3} is proven in Section \ref{sec:3Q}. As a corollary, and combining with the Bogoliubov diagonalization, we obtain the following lower bound on $\mathcal H_\mu(z)$.

\begin{corollary}\label{corollary}
Under the assumptions of Theorem \ref{thm.cnumber} and Theorem \ref{thm.Q3}, for all $|z|^2 \leq K_{\ell}^{1/4}N$ we have
\begin{align*}
 \mathcal H_\mu(z) + \varepsilon \mathcal G(z) &\geq 4\pi |z|^2 \rho_z a \cdot \frac{128}{15\sqrt{\pi}} \sqrt{\rho_z a^3} + \sum_{p \in {\Lambda}^*_+} \tilde{D}_p(z) b_p^* b_p - \mu |z|^2 \\ &\qquad+8\pi a \Big(\frac{\vert z\vert^4}{\vert \Lambda\vert }-\rho N\Big) (\rho a^3)^{\frac{1}{4}} - C \ell^3 (\rho a)^{5/2} K_\ell^{-1/4}
\end{align*}
where 
\begin{align}\label{eq:D_tilde}
\tilde{D}_p(z) = \begin{cases}
D_p(z) &\text{if} \quad p \notin \mathcal P_H,\\
K_H^{-1} D_p(z) &\text{if} \quad p \in \mathcal P_H.
\end{cases}
\end{align}
\end{corollary}

\begin{proof}
We start from \eqref{def.Hmu} and diagonalize the main quadratic part using (\ref{eq.diag}) and then bound the last two terms using Theorem \ref{thm.Q3}. This results in the bound
\begin{align*}
\mathcal H_\mu(z) + \varepsilon \mathcal G(z) &\geq  \frac{1}{2} \sum_{p \in {\Lambda}^*_+} \Big( \sqrt{\tau(p)^2 + 2 \rho_z \widehat g(p) \tau(p)} - \tau(p) - \rho_z \widehat g(p)  \Big) + \frac{1}{2 \vert \Lambda \vert}  \widehat{g \omega}(0) \vert z\vert^4\\ &\qquad  + \sum_{p \in {\Lambda}^*_+} \tilde{D}_p(z) b_p^* b_p -\mu \vert z\vert^2 +\widehat{g}(0) (\rho a^3)^{\frac{1}{4}}\Big(\frac{\vert z\vert^4}{\vert \Lambda\vert }-\rho N\Big) -CN\rho a \sqrt{\rho a^3} K_\ell^{-1}.
\end{align*}
We then approximate the sum by an integral according to Lemma \ref{lem.integral}, and the largest error is of order $\ell^3 (\rho a)^{5/2} K_\ell^{-1/4}$. 
\end{proof}

%%%%%%%%%%%%%%%%%%%%%%%%%%%%%%%%%%%%%%%%%%%%%%%%%%
%%%%%%%%%%%%%%%%%%%%%%%%%%%%%%%%%%%%%%%%%%%%%%%%%%

\subsection{Proof of Theorem \texorpdfstring{\ref{main.thm}}{1.3}}

We make the following choice for the parameters, recalling that $K_\ell = (\rho a^3)^{-\eta}$,
\begin{align}
\label{eq:params}
K_H = K_\ell^{5}, \qquad \gamma =20 \eta, \qquad \alpha = \frac 14 + \frac{\eta}{2}, \qquad \mathcal M = \rho \ell^3 K_\ell^{-21}, \qquad m = 10 \eta^{-1}.
\end{align}
In particular we have $\alpha + \frac{5 \nu}{2} < \frac{6}{17}$ for all $\nu < \frac{\eta}{3}$ and $\eta < \frac{1}{1026}$.

\smallskip
\textbf{Case $N \leq (\rho a^3)^{\alpha} \rho \ell^3$.} Discarding the interaction for a lower bound, we find similarly as in (\ref{eq:bound_ideal}),
\begin{align}
F(\ell,N) 
	&\geq T \sum_{p \in \Lambda^*_+} \log (1-e^{-\frac{p^2}{T}}) \geq T \sum_{p \in \Lambda^*_+} \log \left(1- e^{\tfrac{-1}{T}\sqrt{p^4+16\pi a N \ell^{-3}p^2}  } \right) -  C \frac{a N}{T \ell^3} T^{5/2}\ell^3 \nonumber \\
	&\geq F_{\mathrm{Bog}}(\ell,N) - C \ell^3 (\rho a)^{5/2} \left((\rho a^3)^{2\alpha-1/2} + (\rho a^3)^{\alpha-3\nu/2}  \right), \label{eq:small_N}
\end{align}
where we used \cite[Eq. (8.16)]{HHNST} to estimate the difference between the two sums. Note that we need $\alpha>1/4$ for the error term to be subleading compared to LHY order. We choose $\alpha = 1/4 + \eta/2$ where we recall that $K_\ell = (\rho a^3)^{-\eta}$, so that $(\rho a^3)^{2\alpha-1/2} + (\rho a^3)^{\alpha-3\nu/2} = K_\ell^{-1} + (\rho a^3)^{1/4 + \eta/2 -3\nu/2} \leq C K_\ell^{-1}$.

\smallskip
\textbf{Case  $N > (\rho a^3)^{\alpha} \rho \ell^3$.} We combine Theorems \ref{thm.gaps}, \ref{thm.sym} and \ref{thm.cnumber}, to obtain the lower bound
\begin{equation}\label{eq. first conclusion bound}
	F(\ell,N)\geq 4\pi a \frac{N^2}{|\Lambda|} -T\log  \big[\int_{\mathbb{C}} \tr_{\mathscr{F}^{\perp}}\Big(e^{-\frac{1}{T} (\mathcal{H}_{\mu}(z) + \varepsilon \mathcal G(z))}\Big) e^{-\frac{\mu}{T}N} \dd z\big] - \mathcal E.
\end{equation}
with 
\begin{align*}
\mathcal E \leq C\ell^3 (\rho a)^{5/2} (\rho a^3)^{1/18-2\gamma - \nu} K_H^{3}+ C K_\ell^2 \rho a,
\end{align*}
which holds for $K_H \geq C K_\ell^4$, $K_\ell K_H^3 \leq (\rho a^3)^{- \frac 1 2}$, $ K_H \leq \sqrt{\ell / a} = (\rho a^3)^{-1/4-\eta /2}$ and since $\mathcal M \leq C^{-1} \rho \ell^3 K_H^{-3} K_\ell^{-17/4}$. These conditions are satisfied with our choice of parameters, and the error is $\mathcal E \leq C\ell^3 (\rho a)^{5/2} K_\ell^{-1}$.

Let us decompose the integral in (\ref{eq. first conclusion bound}) as $\int_{\mathbb{C}} = \int_{|z|^2\leq K_{\ell}^{1/4}N} + \int_{|z|^2> K_{\ell}^{1/4}N} = X + Y$. We will use that if $-T \log X \geq Z$ and $-T \log Y \geq Z$, then 
\begin{align}\label{eq:log_story}
-T \log (X + Y) \geq -T \log (2e^{-Z/T}) = Z - T \log 2.
\end{align}
For the relevant region of $z$, where $|z|^2 \leq K_{\ell}^{1/4}N$, we can use Corollary \ref{corollary}, which gives
\begin{align*}
 \mathcal{H}_{\mu}(z) + \varepsilon \mathcal G(z) + \mu N  &\geq 4\pi |z|^2 \rho_z a \cdot \frac{128}{15\sqrt{\pi}} \sqrt{\rho_z a^3} + \sum_{p\in \Lambda^*_+} \tilde{D}_p(z) b_p^* b_p - \mu |z|^2 +  8\pi a(\rho a^3)^{\frac{1}{4}}\rho_z\vert z\vert^2 \\
  &\qquad +  N (\mu - 8 \pi \rho a (\rho a^3)^{\frac{1}{4}}) - \mathcal{E}'
\end{align*}
with $$\mathcal{E}'\leq C\ell^3 (\rho a)^{5/2}K_\ell^{-1/4}.$$
Now using Lemma \ref{lem.thermal.approx}, we obtain
\begin{equation*}
-T \log \tr_{\mathscr{F}^{\perp}} \Big( e^{-\frac 1 T \sum_{p \in \Lambda^*_+} \tilde{D}_p(z) b_p^* b_p} \Big) 
	= T\sum_{p \in \Lambda^*_+} \log \big( 1 - e^{- \frac 1 T \tilde{D}_p(z)} \big) \geq T\sum_{p \in \Lambda^*_+}\log(1-e^{-\frac{1}{T}\omega_p(z)})-  C \ell^3 (\rho a)^{3}
\end{equation*}
and denoting $\omega_p(z) = \sqrt{p^4 + 16 \pi a |z|^2\ell^{-3} p^2}$, we have
\begin{align*}
-T \log \tr_{\mathscr F^\perp} \left(e^{-\frac{1}{T}\left(\mathcal{H}_{\mu}(z) + \varepsilon \mathcal G(z) + \mu N\right)}\right) \geq F(\vert z\vert ^2) +  N (\mu -8 \pi \rho a (\rho a^3)^{\frac{1}{4}}) - \mathcal{E}'
\end{align*}
with
\begin{equation}\label{eq:Fz}
F(\vert z\vert ^2):= 8\pi a(\rho a^3)^{\frac{1}{4}}\rho_z\vert z\vert^2 +  4\pi |z|^2 \rho_z a \cdot \frac{128}{15\sqrt{\pi}} \sqrt{\rho_z a^3} + T\sum_{p\in \Lambda^*_+}\log(1-e^{- \frac 1 T \omega_p(z)})  -\mu \vert z\vert^2,
\end{equation}
By Lemma \ref{lem:convexity2}, we deduce that $F$ is convex for $\rho a^3$ and $\eta$ small enough. Choosing $\mu$ so that $F'(N) = 0$, $F$ achieves its minimum at $\vert z\vert^2=N$. We obtain
\begin{align}
	&-T \log \int_{\vert z\vert^2 < K_\ell^{1/4} N} \tr_{\mathscr{F}^{\perp}} \Big(e^{-\frac{1}{T} (\mathcal{H}_{\mu}(z) + \varepsilon \mathcal G(z))}\Big) e^{-\frac{\mu}{T}N} \dd z  \nonumber  \\
		&\geq 4\pi a \frac{N^2}{|\Lambda|} \cdot \frac{128}{15\sqrt{\pi}} \sqrt{ \frac{N a^3}{|\Lambda|}}  + T\sum_{p\in \Lambda^*_+}\log(1-e^{- \frac 1 T \omega_p(\sqrt{N})}) 
		- \mathcal{E}' - T \log (CK_{\ell}^{1/4} N). \label{eq. energy for small z}
\end{align}
On the other hand when $|z|^2 \geq K_{\ell}^{1/4} N$, and since our choice of $m$ satisfies both $m > 2\eta^{-1}+14$ and $K_H \leq K_{\ell}^{\frac{m+1}{12}}$, we use Lemma \ref{lem.largez} to obtain
\begin{align*}
- T \log &\int_{\vert z\vert^2 \geq K_{\ell}^{1/4}N} \tr_{\mathscr{F}^{\perp}} \Big(e^{-\frac{1}{T} (\mathcal{H}_{\mu}(z) + \varepsilon \mathcal G(z))}\Big) e^{-\frac{\mu}{T}N} \dd z \\
&\geq  T\sum_{p\in \Lambda^*_+}\log(1-e^{- \frac{\tau(p)}{2T}}) + \mu N - T \log \int_{|z|^2 \geq K_{\ell}^{1/4}N} \exp\Big(-   \frac{c\rho a |z|^{2}}{TNK_{\ell}}\Big) \dd z \\
&\geq  T\sum_{p\in \Lambda^*_+}\log(1-e^{- \frac{p^2}{4T}}) + \mu N - T \log \bigg(\frac{CTN}{\rho a K_{\ell}^{\frac{m-1}{4}}} \exp\Big(- \frac{c \rho a K_{\ell}^{\frac{m}{4}}}{T}\Big) \bigg)\\
&\geq - CT^{5/2}\ell^3 - T \log \left(\frac{CTN}{\rho a K_{\ell}^{\frac{m-1}{4}} }\right)  + c \rho a K_{\ell}^{\frac{m}{4}},
\end{align*}
for $\rho a^3$ small enough. We used that $\tau(p) \geq p^2/2$ and that $\mu \geq 0$.
%
%}
%
The above is clearly bigger than the right-hand side of (\ref{eq. energy for small z}) for $\rho a^3$ small enough due to the constraint on $N$ and the assumptions on $m$. Thus for $\rho a^3$ small enough, using (\ref{eq:log_story}), we obtain
\begin{align*}
	-T\log  \big[\int_{\mathbb{C}} \tr_{\mathscr{F}^{\perp}}\Big(e^{-\frac{1}{T} (\mathcal{H}_{\mu}(z) + \varepsilon \mathcal G(z))}\Big) e^{-\frac{\mu}{T}N} \dd z\big] 
		&\geq 4\pi a \frac{N^2}{|\Lambda|} \cdot \frac{128}{15\sqrt{\pi}} \sqrt{ \frac{N a^3}{|\Lambda|}}  + T\sum_{p\in \Lambda^*_+}\log(1-e^{- \frac 1 T \omega_p(\sqrt{N})}) \nonumber \\
		&\quad - \mathcal{E}' - T \log (CK_{\ell}^{1/4} N).
\end{align*}
We combine the above with \eqref{eq. first conclusion bound} to obtain for $N > (\rho a^3)^{\alpha} \rho \ell^3$
\begin{equation}\label{eq:F_bound_high}
	F(\ell,N)\geq 4\pi a \frac{N^2}{|\Lambda|} \Big(1+\frac{128}{15\sqrt{\pi}} \sqrt{\frac{N a^3}{|\Lambda|}} \Big)+T\sum_{p \in \Lambda^*}\log(1-e^{-\frac{1}{T}\omega_p(\sqrt{N})})-C T\log(N)-\mathcal{E}-\mathcal{E}'.
\end{equation}
We have
\begin{align*}
C T\log(N)+\mathcal{E}+\mathcal{E}' 
	&\leq C \ell^3 \Big((\rho a)^{5/2} (\rho a^3)^{1/18-2\gamma - \nu} K_H^{3} + \frac{T}{\ell^3} \log N  + (\rho a)^{5/2} K_\ell^{-1/4} \Big) \\
	&\leq C \ell^3 (\rho a)^{5/2} \left( K_\ell^{-1/4} + (\rho a^3)^{-\nu} K_\ell^{-3} |\log (\rho a^3)|  \right) \\
	&\leq C \ell^3 (\rho a)^{5/2} K_\ell^{-1/4}.
\end{align*}
Combining the cases $N \leq (\rho a^3)^{\alpha} \rho \ell^3$ and  $N > (\rho a^3)^{\alpha} \rho \ell^3$ finishes the proof of \Cref{main.thm}.

%%%%%%%%%%%%%%%%%%%%%%%%%%%%%%%%%%%%%%%%%%%%%%%%%%
%%%%%%%%%%%%%%%%%%INTPOTENTIALS%%%%%%%%%%%%%%%%%%%%%%%%
%
\section{Approximation by integrable potentials}\label{sec.scattering}
This section is devoted to the proof of Proposition~\ref{prop.L1v}.
We begin by recalling the definition of the scattering length and related quantities. For more details see \cite{greenbook}.
\begin{definition}
	Let $V: \mathbb{R}^{3} \to \mathbb{R}^+\cup\{\infty\}$ be measurable and radial with support in $B(0,R)$. The scattering length $a=a(V)$ is defined as
	\begin{equation}\label{variational definition of the scattering length}	4\pi a=\inf\Big\{\int\vert \nabla \varphi\vert^2+\frac{1}{2}V|\varphi|^2 \dd x\, \Big\vert \, \varphi\in \dot{H}^1(\mathbb{R}^3),\quad \lim_{\vert x\vert \rightarrow \infty}\varphi(x)=1\Big\},
	\end{equation}
\end{definition}

An important special case is the hard core potential of radius $R>0$
\begin{align}
V_{\rm hc}(x) := \begin{cases} +\infty ,& |x| \leq R, \\ 0, & |x| > R. 
\end{cases}
\end{align}
For this special potential it is not difficult to see that $a(V_{\rm hc})=R$. For general potentials, 
inserting the test function $\max\{0, 1-\frac{R}{|x|}\}$, we find $a\leq R$. It is also easy to verify that $a$ is an increasing function of $V$. The minimizer $\varphi_V$  solves the corresponding Euler-Lagrange equation
\begin{equation} \label{scattering_equation}
	-\Delta \varphi_V+\frac{1}{2}V\varphi_V =0,
\end{equation}
in a weak sense. It is easy to check with Newton's theorem that
\begin{equation}\label{eq.simple form of phi}
	\varphi_V(x)=1-\frac{a}{\vert x\vert}, \quad \text{for $\vert x\vert \geq R$},
\end{equation}
and furthermore $\varphi_V$ is non-decreasing, non-negative and radial. 
We will also use the following standard monotonicity result from \cite[Lemma~C.2]{greenbook}: if $V_1 \geq V_2 \geq 0$, then
$\varphi_{V_1}(x) \leq \varphi_{V_2}(x)$ for all $x$.
We will omit the $v$ from the notation of the scattering length and write $\varphi:=\varphi_V$, if the potential is clear from the context.
We then recall the notation
\begin{align}\label{eq:scat_defs}
	\omega = 1- \varphi,
	\qquad g = V \varphi=V(1-\omega).
\end{align}
Clearly,
\begin{align}\label{eq:ScatOmega}
	-\Delta \omega = \frac{1}{2} g, \qquad \text{and} \qquad 
	\widehat{g}(0)=	\int g \,\dd x = 8\pi a.
\end{align}
Having introduced the necessary theory and notation we may provide the proof of Proposition~\ref{prop.L1v}. The proof will rely on the following two lemmas.

\begin{lemma}\label{adiff}
For all $v_1 \geq v_2 \geq 0$ and $v'\geq 0$ we have
\[a(v_1) - a(v_2) \geq a(v_1+v') - a(v_2+v').\]
\end{lemma}

\begin{proof}
For all $t \in [0,1]$, we introduce $\varphi_1^t$ and $\varphi_2^t$ as the scattering solutions for $v_1+tv'$ and $v_2+tv'$ respectively. By definition we have for $j=1,2$,
\[4 \pi a(v_j + t v') = \int \Big(| \nabla \varphi_j^t |^2 + \frac{1}{2} (v_j + t v') |\varphi_j^t|^2 \Big)\dd x.\]
In particular, rearranging the terms we find
\begin{align*} 
4\pi a(v_1+tv') - 4\pi a(v_2+tv') &=   \int \Big(\nabla \varphi^t_1 \cdot \nabla (\varphi_1^t - \varphi_2^t) + \frac{1}{2} (v_1 + t v') \varphi_1^t (\varphi_1^t - \varphi_2^t) \Big)\dd x \\  &\qquad - \int \Big(\nabla \varphi^t_2 \cdot \nabla (\varphi_2^t - \varphi_1^t) + \frac{1}{2} (v_2+t v') \varphi_2^t (\varphi_2^t - \varphi_1^t)\Big) \dd x \\ &\qquad +\frac{1}{2} \int (v_1 - v_2) \varphi^t_1 \varphi^t_2 \,\dd x.
\end{align*}
When we integrate by parts, we find that the two first lines vanish by \eqref{scattering_equation}. For instance,
\[ \int \Big(\nabla \varphi^t_1 \cdot \nabla (\varphi_1^t - \varphi_2^t) + \frac{1}{2} (v_1+tv') \varphi_1^t  (\varphi_1^t - \varphi_2^t)\Big) \dd x = \int (-\Delta \varphi_1^t + \frac 1 2 (v_1+tv') \varphi_1^t) (\varphi_1^t - \varphi_2^t) \dd x = 0. \]
Therefore,
\[a(v_1+tv') -  a(v_2+tv') = \frac{1}{8\pi }  \int (v_1 - v_2) \varphi^t_1 \varphi^t_2 \dd x,\]
for all $t \in [0,1]$. Comparing with $t=0$ we deduce
\[ a(v_1+tv') - a(v_2+tv') =  a(v_1) -  a(v_2) + \frac{1}{8\pi}  \int (v_1 - v_2) ( \varphi^t_1 \varphi^t_2 - \varphi_1^0 \varphi_2^0) \dd x.\]
Since $v_j+tv' \geq v_j$, we have $0 \leq \varphi_j^t \leq \varphi_j^0$ pointwise, and thus the last integral is negative. The result follows with $t=1$.
\end{proof}

\begin{lemma}\label{lem. cutting head of potential}
	Given $K>0$ and a non-increasing potential $v$, the potential $\min(v,K)$ satisfies
	\[a(V) \geq a(\min(V,K))\geq a(V)-\frac{2\sqrt{2}}{\sqrt{K}}\]
\end{lemma}
\begin{proof}	
Using that $V$ is decreasing, we get $\{V>K\}=B(0,R_K)$ for some $R_K \geq 0$.
Note that the hard core potential with radius $R_K$ has scattering length $R_K$.
	Then comparing to the hard core potential, we have
	\begin{equation}\label{eq. bound on scattering length}
		a(V)-a(\min(V,K))\leq a(V \one_{B(0,R_K)})-a(K\one_{B(0,R_K)})\leq R_K-a(K\one_{B(0,R_K)}),
	\end{equation}
where in the first inequality we used Lemma \ref{adiff} with $V' = V \one_{B(0,R_K)^c}$, and in the second inequality that $V \one_{B(0,R_K)}$ is smaller than the hardcore potential $V_{hc}$ of radius $R_K$.
	 Now one can compute $a_K:=a(K\one_{B(0,R_K)})$ by solving
	\[-\Delta \varphi +\frac{1}{2}K\one_{ B(0,R_K)}\varphi=0,\]
	to find that there is a $c>0$ such that
	\[\varphi(x)=\begin{cases}
		c\frac{\sinh(\sqrt{\frac{K}{2}}\vert x\vert)}{\vert x\vert},\qquad &\vert x\vert \leq R_K,\\
		1-\frac{a_K}{\vert x\vert},\qquad &\vert x\vert \geq R_K.
	\end{cases}\]
	Knowing that $\varphi$ is continuous and differentiable yields we find that
	\begin{equation}\label{eq. formula for scattering length}
		R_K-a_K=\frac{1-e^{-2\gamma}}{\gamma-1+e^{-2\gamma}(\gamma+1)} a_K\qquad \text{where}\quad \gamma=\sqrt{\frac{K}{2}}R_K.
	\end{equation}
	Lastly we observe that for $\gamma\geq 2$ the right-hand side of \eqref{eq. formula for scattering length} is bounded by $\frac{2}{\gamma} a_K$, and if $\gamma\leq 2$ we have $R_K\leq \frac{2\sqrt{2}}{\sqrt{K}}$. Thus combining \eqref{eq. bound on scattering length} with \eqref{eq. formula for scattering length}, using $\frac{a_K}{R_K}\leq 1$, we get
	\[a(V)-a(\min(V,K))\leq R_K-a_K\leq \frac{2\sqrt{2}}{\sqrt{K}}.\]
\end{proof}
Before giving the proof of Proposition~\ref{prop.L1v}, we explain the construction of $v$ in the case where $V=V_{\rm hc}$ is the hard core potential of radius $a$. In \cite{FS2} it was explained that to get \eqref{eq.propv1} the (almost optimal) approximation of $V_{\rm hc}$ is by a `thin shell' potential supported on the annulus $A:=\{ a- a^2 \ell^{-1} \leq \vert x \vert \leq a\}$ and of height $\ell^2 a^{-4}$.
However, this potential clearly doesn't satisfy \eqref{eq.propv2}, since it vanishes inside the shell. The remedy is to fill the inside of the shell, without changing too much the $L^1$ norm. Therefore, the final choice of $v$ is
\[
v = \ell a^{-3} \one_{B(0,a- a^2\ell^{-1})} + \ell^2 a^{-4} \one_{A}, \qquad \text{when $ V = V_{hc}$ is the hard core.}
\]
Let us now give the details starting from an arbitrary $V$. For convenience we include the following lemma.

\begin{lemma}\label{lem:hardcorereg}
There exists a universal constant $C_0>0$ such that the following is true.
Suppose that $V: {\mathbb R}^3 \rightarrow  [0,\infty]$ is radial and of class $L^1$ with compact support of radius $R$ and that $S\geq 0$. 
If $S \geq \frac{\int V}{8\pi a(V)}$ define the potential $V_S:= V$.
If $S < \frac{\int V}{8\pi a(V)}$ define
\[ V_S := V \one_{[R_S, \infty)}, \]
with $R_S$ chosen so that $\int V_S = 8 \pi S a(V)$. Then $V_S$ satisfies
\[
a(V) \geq a(V_S) \geq a(V) (1 - C_0 S^{-1}).
\]
\end{lemma}

This is a reformulation of \cite[Lemma 3.3]{FS2} using the fact that the proof gives the explicit construction of $v_S$. Using this Lemma we can prove Proposition \ref{prop.L1v}.

\begin{proof}[Proof of Proposition~\ref{prop.L1v}]
By Lemma \ref{lem. cutting head of potential}, we may assume that $V \leq K$, with $K = \ell^2 a^{-4}$. This choice guarantees that the change to the scattering length is of order $a^2 \ell^{-1}$. 
By Lemma~\ref{lem:hardcorereg}, we can find $0<R_S<R$ such that the potential
	\[V_S(x)=\begin{cases}
		V(x)\qquad &\vert x\vert \geq R_S\\
		0\qquad &\vert x\vert< R_S
	\end{cases}\]
	has integral $\int V_S \leq 8 \pi S a(V)$ and satisfies 
\begin{equation}\label{eq:tail}
a(V) \geq a(V_S) \geq a(V) (1 - C_0 S^{-1} ).
\end{equation}
Note that this error is small enough with the choice $S=\ell / a(V)$. However note that $V_S$ does not satisfy \eqref{eq.propv2}, so it requires further modifications. First, we extend slightly $V_S$ and define
\begin{equation}
w_S = V_S(R_S) \one_{[R_S - \varepsilon, R_S]} + V_S,
\end{equation}
where we used the convention $\one_I(x) = \one_{\{ |x| \in I \}}$ for a subset $I \subset \mathbb R$, and where  $\varepsilon >0$ will be chosen later.
Let $g_S = w_S \varphi_{w_S}$, and let $x_0$ be the maximal point of $g_S$,
\[ g_S(x_0) = \sup g_S.\]
Note that $x_0 \geq R_S$ by construction.
Then we define our potential $ v$ by
\begin{equation}
 v:= \min \big( g_S(x_0), M \big) \one_{[0,R_S - \varepsilon]} + w_S,
\end{equation}
where the parameter $M$ will be chosen to bound the $L^1$ norm of $v$. Indeed, choosing $\varepsilon = a(V)^2 \ell^{-1}$ and $M = \ell R_S^{-3}$,
\begin{align*}
 \int v \leq C R_S^3 M + C   a(V)^2 \ell^{-1} R_S^2  V_S(R_S) + \int V_S \leq C \ell,
\end{align*}
where we used that $V_S(R_S) \leq K = \ell^2 a(V)^{-4}$, $R_S \leq R \leq C a(V)$ and $\int V_S \leq 8\pi \ell$.  We can compare the scattering lengths of $v$ and $v$ using that $V \geq  v \geq V_S$ and therefore, we get from \eqref{eq:tail} that
 \[
 a(V) \geq a(v) \geq a(V) (1 - C a(V)/\ell ).
 \]

Let us now show that the potential $v$ satisfies \eqref{eq.propv2}. We consider $|x|< |y|$.
\begin{itemize}
\item If $|x| \geq R_S - \varepsilon$ then $v$ is non-increasing on this region, and therefore since $\varphi_{v} \leq 1$ we have
\[  v(x) \geq  v(y) \geq  v(y) \varphi_{ v}(y) = g_v(y).\]
\item When $|x| \leq R_S - \varepsilon$.
Notice that $ v \varphi_{ v}$ is increasing on $[0, R_S]$. Therefore,
\begin{align}\label{gybound}
 g_v(y) \leq  \sup_{|x|\geq R_S}  v \varphi_{ v} \leq \sup_{|x|\geq R_S} w_S \varphi_{w_S} = g_S(x_0),
\end{align}
where the last inequality follows since $\varphi_{w_S} \geq \varphi_{ v}$. If $g_S(x_0) < M$ then $g_S(x_0) =  v(x)$ and we are done. In the other case, when $M \leq g_S(x_0)$ we introduce an auxiliary potential,
\begin{equation}
v_0 = v(x_0) \one_{[R_S-\varepsilon,|x_0|]}.
\end{equation}
We observe that $v_0 \leq w_S$ and this implies that $\varphi_{w_S} \leq \varphi_{v_0}$. Therefore
\begin{equation}\label{gSbound}
g_S(x_0) \leq V(x_0) \varphi_{v_0}(x_0) = V(x_0) \Big( 1 - \frac{a(v_0)}{|x_0|} \Big).
\end{equation}
We estimate the difference between $a(v_0)$ and $|x_0|$ identified with the scattering length of the hard-core of radius $|x_0|$. We use Lemma \ref{lem. cutting head of potential} with $K = v(x_0)$ to cut the hard-core potential, and then compare the resulting potential to $v_0$ using Lemma~\ref{lem:hardcorereg} (with $8 \pi S |x_0|= \int v_0$) to get
\[ |x_0| - a(v_0) \leq  \frac{C}{\sqrt{v(x_0)}} +  C \frac{|x_0|^2}{\int v_0} .\]
Inserting this bound in \eqref{gSbound} we get
\[ g_S(x_0) \leq C \frac{\sqrt{v(x_0)}}{|x_0|} + C a(V) \frac{V(x_0)}{\int v_0} \leq C \frac{\ell}{ a(V)^2 R_S} + C \frac{a(V)}{\min(\varepsilon, R_S) R_S^2} \leq C M,\]
where we used $V(x_0) \leq \ell^2 a(V)^{-4}$, $\varepsilon = a(V)^2/\ell$, and $R_S \leq |x_0| \leq R \leq C a(V)$.
Therefore,
\[g_S(x_0) \leq C M = C  v(x)\]
which combined with \eqref{gybound} gives the desired bound.
\end{itemize}
\end{proof}

%%%%%%%%%%%%%%%%%%%%%%%%%%%%%%%%%%%%%%%%%%%%%%%%%%
%%%%%%%%%%%%%%%LOCALISATIONLARGEMATRICES%%%%%%%%%%%%%%%%%%%%%%%

\section[Localization of large matrices]{Localization of large matrices: Proof of Theorem \ref{thm.gaps}} \label{sec.gaps}

One of the main ingredients to prove Theorem \ref{thm.gaps} is a rough condensation estimate for low energy states. This is by now a well known result and its proof can be found in \cite{greenbook}, and the technique was also crucial in \cite{HHNST}. At positive temperature, we need to extend this property to mixed states.

\begin{lemma}[Condensation of low energy states]\label{lem.condensation}
There exists a $C>0$ such that the following holds. Assume that $\rho \ell^{3} (\rho a^3)^{\alpha} \leq N \leq 20 \rho \ell^3$, $T \leq \rho a (\rho a^3)^{-\nu}$, with $\alpha + 5\nu/2 < 6/17$  and $(\rho a^3)\leq C^{-1}$. Let $\Gamma$ be a trace-class operator on $L^2(\Lambda^N)$ such that $\Gamma \geq 0$, $\rm{Tr}\, \Gamma =1$, and
\begin{equation}\label{assmpt:lowE}
\mathrm{Tr}(H_N \Gamma) \leq 4 \pi N^2 \ell^{-3}  a(1 + (\rho a^3)^{\frac{1}{17}}).
\end{equation}
Then we have 
\begin{align}
\mathrm{Tr}(n_+\Gamma) &\leq C N K_{\ell}^2(\rho a^3)^{\frac{1}{17}} . \label{eq:apriori_n}
\end{align}
In particular, this holds for the Gibbs state
\begin{equation}
\Gamma_0 = \frac{e^{-\frac{H_N}{T}}}{\mathrm{Tr}(e^{-\frac{H_N}{T}})}.
\end{equation}
\end{lemma}

\begin{proof} The result follows from the following lower bound, which holds as long as $N (\rho a^3)^{\frac{1}{17}} \geq 1$,
\begin{equation}
H_N \geq 4 \pi a \frac{N^2}{\ell^3}  (1 - C   (\rho a^3)^{\frac{1}{17}}) + C' \frac{n_+}{\ell^2},
\end{equation}
and which can be found in \cite[Lemma 4.1 and Lemma 5.2]{greenbook}. The constraint on $\alpha$ ensures in particular the condition on $N$. Together with the upper bound \eqref{assmpt:lowE} we deduce
\begin{equation}
{\rm{Tr}} \Big( \frac{n_+}{\ell^2} \Gamma \Big) \leq C N \rho a (\rho a^3)^{\frac{1}{17}},
\end{equation}
which is the expected condensation estimate. It remains to show that $\Gamma_0$ satisfies the upper bound \eqref{assmpt:lowE}, which is not obvious due to the entropy term. To prove this, we use the upper bound on the free energy from \cite{greenbook},
\begin{equation}\label{eq.anupperbound}
\mathrm{Tr}(H_N \Gamma_0 + T\, \Gamma_0 \ln \Gamma_0) \leq  \inf \sigma (H_N) \leq 4 \pi \rho N a (1+ C (\rho a^3)^{1/3}).
\end{equation}
The upper bound on the ground state energy was proven by Dyson \cite{dyson} (see also \cite{newupperbound} for an improvement to the order of the LHY correction, but here we do not need such precise estimates).
Therefore, using \eqref{eq.anupperbound} and the Gibbs variational principle we can bound the energy for any $0<\varepsilon<1$, 
\begin{align} \label{eq:bound_ideal}
\mathrm{Tr}(H_N \Gamma_0) &\leq (1+\varepsilon) \mathrm{Tr}(H_N \Gamma_0+ T \, \Gamma_0 \ln \Gamma_0) - \big( \varepsilon \mathrm{Tr}(-\Delta_{\mathbb{R}^{3N}} \Gamma_0) + (1+\varepsilon)T  \, {\rm{Tr}}  (\,\Gamma_0 \ln \Gamma_0 ) \big)	\nonumber \\	
&\leq (1+\varepsilon) 4\pi \rho N a (1 + C (\rho a^3)^{1/3}) + (1+\varepsilon)T \log \mathrm{Tr} \big( e^{-\frac{\varepsilon}{(1+\varepsilon)T}(-\Delta_{\mathbb{R}^{3N}})}\big),
\end{align}
In the second term, the free energy of the ideal gas is bounded by
\begin{equation*}
T \log \mathrm{Tr}_{\mathscr F^{\leq N}_+} \big( e^{-\frac{\varepsilon}{(1+\varepsilon)T}\sum_{p\neq 0} p^2 a^*_pa_p} \big)  \leq  T \sum_{p \in \ell^{-1}\mathbb{N}^3 \setminus \{0\}} \log \big( 1-e^{-\frac{\varepsilon}{(1+\varepsilon)T}p^2}\big) \leq C \varepsilon^{-3/2}T^{5/2} \ell^3,
\end{equation*}
where $\mathscr F^{\leq N}_+ = \bigoplus_{n=0}^N (\Ran Q)^{\otimes n} \simeq L^2(\Lambda^N)$ is the truncated bosonic Fock space of excitations. Using that $T \leq C \rho a (\rho a^3)^{-\nu}$ and $N \geq \rho \ell^3 (\rho a^3)^{\alpha}$, we obtain
\begin{align*}
\mathrm{Tr}(H_N \Gamma_0) 
	&\leq 4\pi \rho N a (1 + C (\rho a^3)^{1/3} + C (\varepsilon + \varepsilon^{-3/2} \frac{T^{5/2} \ell^3}{N \rho a})) \\
	&\leq 4\pi \rho N a (1 + C (\rho a^3)^{1/3} + C (\varepsilon + \varepsilon^{-3/2} (\rho a^3)^{1/2-\alpha - 5\nu/2})) \\
	&\leq 4\pi \rho N a (1 + C (\rho a^3)^{1/3} + C (\rho a^3)^{(1-2\alpha-5\nu)/5}),
\end{align*}
where we optimized in $\varepsilon$ to obtain the last inequality. The condition that $\alpha + 5\nu/2 < 6/17$ leads to the estimate \eqref{assmpt:lowE}.
\end{proof}
Theorem \ref{thm.gaps} follows from Proposition~\ref{thm:excitationrestriction} below, together with the a priori bounds on $n_+$ from Lemma~\ref{lem.condensation}. The proof of Proposition~\ref{thm:excitationrestriction} is inspired by the localization of large matrices as in \cite{FS2}, and simplified in \cite{2DLHY}. It is also similar to the bounds in \cite[Proposition 21]{HST}. It can be interpreted as an analogue of the standard IMS localization formula. It roughly says that, for a lower bound, we can restrict the estimates states $\Gamma_{\mathcal M}$ which have bounded number of low excitations.
\begin{proposition}[Localization to $\{n_+^L\leq \mathcal{M}\}$]\label{thm:excitationrestriction}
Let $\Gamma_0$ be the Gibbs state associated to $H_N$. Let $\mathcal{M} \geq N (\rho a^3)^\gamma$, for some $\gamma>0$, then there exists a trace class operator $\Gamma_{\mathcal{M}}$ on $L^2(\Lambda^N)$ with $ \Gamma_{\mathcal M} \geq 0$ and trace $1$ such that
\begin{equation}\label{eq. the condition on gamma}
	\one_{\{n_+^L\leq \mathcal{M}\}}\Gamma_{\mathcal{M}}\one_{\{n_+^L\leq \mathcal{M}\}}=\Gamma_{\mathcal{M}},
\end{equation}
and
\begin{align*}
F(\ell,N) &\geq \mathrm{Tr}\big( H_N\Gamma_{\mathcal{M}} + T \, \Gamma_{\mathcal{M}} \ln \Gamma_{\mathcal M})  - C \frac{1}{\mathcal M^2} \left(\tr ( H_N \Gamma_0 )+ \|v\|_1K_H^3 \ell^{-3}  N \tr (n_+ \Gamma_0) \right) \\
	&\qquad  -CT \frac{\tr(n_+ \Gamma_0)}{\mathcal M} \left(1+ \left| \log \frac{\tr(n_+ \Gamma_0)}{\mathcal M}\right|\right).
\end{align*}
\end{proposition}
\noindent Theorem \ref{thm.gaps} follows from this proposition.
\begin{proof}[Proof of Theorem \ref{thm.gaps}]
Using the a priori bounds on $n_+$ from Lemma \ref{lem.condensation}, we have for $\mathcal M \geq N (\rho a^3)^{\gamma}$,
\begin{align}\label{ineq.n+M}
\frac{\tr(n_+ \Gamma_0)}{\mathcal M} \leq (\rho a^3)^{1/17-\gamma} K_\ell^{2}.
\end{align}
We use Proposition \ref{thm:excitationrestriction}, and bound the error terms using \eqref{ineq.n+M} and the upper bound \eqref{assmpt:lowE} on the energy of $ \Gamma_0$,
\begin{align*}
&F(\ell,N) - \mathrm{Tr}\big(  H_N\Gamma_{\mathcal{M}} + T \, \Gamma_{\mathcal{M}} \ln \Gamma_{\mathcal M} \big) \\
	&\geq - C \ell^{-3}\left(a (\rho a^3)^{-2\gamma} + \|v\|_1K_H^3  (\rho a^3)^{1/17-2\gamma} K_\ell^{2}\right) - CT (\rho a^3)^{1/17-\gamma} K_\ell^2 |\log (\rho a^3)| \\
	&\geq - C \ell^3 (\rho a)^{5/2} \left((\rho a^3)^{1/2 - 2\gamma} K_\ell^{-6} + K_H^{3} K_\ell^{-3} (\rho a^3)^{1/17-2\gamma} + (\rho a^3)^{1/17-\gamma-\nu} K_{\ell}^{-1} |\log (\rho a^3)|  \right)
\end{align*} 
where we used that $\ell = K_\ell (\rho a)^{-1/2}$, $T \leq C (\rho a) (\rho a^3)^{-\nu}$ and that $\|v\|_1 \leq \ell$. For $\rho a^3$ small enough and using that $K_\ell \geq 1$, this is smaller than the claimed error term of Theorem \ref{thm.gaps}.
It remains to extract the spectral gaps. By definition of $H_N^{\rm{mod}}$ in \eqref{eq.Hmod} we have
\begin{align*}
\mathrm{Tr}\big( H_N \Gamma_{\mathcal{M}}\big) 
	&= \mathrm{Tr}\Big( \Big(  H_N^{\rm{mod}} + \frac{\pi}{2\ell^2} n_+  + \frac{K_H}{\ell^2} n_+^H \Big) \Gamma_{\mathcal M}\Big) \\
	&\geq \mathrm{Tr}\Big( \Big( H_N^{\rm{mod}}+ \frac{\pi}{4\ell^2} n_++ \frac{K_H}{2\ell^2} n_+^H + \frac{\pi n_+^L n_+}{4\mathcal M \ell^2} + \frac{K_H  n_+^L  n_+^H}{2\mathcal M \ell^2} \Big) \Gamma_{\mathcal{M}} \Big) \\
	&= \mathrm{Tr}\Big(\big( H_N^{\rm{mod}} + G \big) \Gamma_{\mathcal M} \Big).
\end{align*}
This concludes the proof of Theorem \ref{thm.gaps}.
\end{proof}

\subsection{Proof of Proposition~\ref{thm:excitationrestriction}}
The rest of this section is dedicated to the proof of Proposition~\ref{thm:excitationrestriction}. It will follow from the Lemmas~\ref{lem:localization_largeMat} and~\ref{lem:d1d2estimate} below, both of which are adapted from \cite{2DLHY}.

\begin{lemma}\label{lem:localization_largeMat}
	Let $\theta :\mathbb{R}  \rightarrow [0,1]$ be any compactly supported smooth function such that $\theta(s) = 1$ for $\vert s \vert < \frac 1 8$ and $\theta(s) = 0$ for $\vert s \vert > \frac 1 4$. For any $\mathcal M \geq 1$, define $c_{\mathcal M} >0$  and $\theta_{\mathcal M}$ such that
	\[ \theta_{\mathcal M}(s) = c_{\mathcal M} \theta \Big( \frac{s}{\mathcal M} \Big) , \qquad \sum_{s \in \mathbb{Z}} \theta_{\mathcal M}(s)^2 = 1 .\]
	Then there exists a $C>0$ depending only on $\theta$ such that, for any normalized state $\Gamma$,
	\begin{equation}\label{eq:large_matrices_decomposn+}
\mathrm{Tr}(H_N\Gamma)\geq \sum_{m\in \mathbb{Z}} \mathrm{Tr}(H_N \Gamma_m)-\frac{C}{\mathcal{M}^2}(\vert \mathrm{Tr}(d_{1}\Gamma)\vert +\vert \mathrm{Tr}(d_{2}\Gamma)\vert )
	\end{equation}
	where $\Gamma_m = \theta_{\mathcal M}(n_+^L - m) \Gamma \theta_{\mathcal M}(n_+^L - m) $ and
	\begin{align}
	d_1 &=  \sum_{i\neq j} Q^L_i (1-Q^L_j) v(x_i-x_j) [ Q^L_i Q^L_j + (1-Q^L_j) (1-Q^L_j)] + \hc, \label{eq:d1}\\
	d_2 &=  \frac{1}{2} \sum_{i\neq j} Q^L_i Q^L_j v(x_i-x_j) (1-Q^L_i) (1-Q^L_j) + \hc \label{eq:d2}
	\end{align}
\end{lemma}

\begin{proof}
	Notice that $H_N$ only contains terms that change $n_+^L$ by $0, \pm 1$ or $\pm 2$. Therefore, we write our operator as $ H_N = \sum_{\vert k \vert \leq 2} \mathcal{H}^{(k)}$,
	with $ \mathcal H^{(k)} n_+^L = (n_+^L + k) \mathcal H^{(k)}$. Moreover, one easily checks that for $k=1,2$,
	\begin{align*}
	\mathcal H^{(k)} + \mathcal H^{(-k)} = d_k.
	\end{align*}
	For $|k| \leq 2$, we have
	\begin{align*}
	\sum_{m \in \mathbb{Z}} \mathrm{Tr}(  \mathcal H^{(k)} \Gamma_m  ) 
		&= \sum_{m \in \mathbb{Z}} \mathrm{Tr}( \theta_{\mathcal M}(n_+^L-m)  \mathcal H^{(k)} \theta_{\mathcal M}(n_+^L-m)   \Gamma) \\
		&= \sum_{m \in \mathbb{Z}} \mathrm{Tr}( \theta_{\mathcal M}(n_+^L -m)  \theta_{\mathcal M}(n_+^L-m+k)   \mathcal H^{(k)}  \Gamma) \\
		&= \sum_{m \in \mathbb{Z}} \theta_{\mathcal M}(m)  \theta_{\mathcal M}(m+k) \mathrm{Tr}( \mathcal H^{(k)}  \Gamma),
	\end{align*}
	where we used the commutation property of $\mathcal H^{(k)}$ and that the function $\sum_{m \in \mathbb{Z}} \theta_{\mathcal M}(X-m)  \theta_{\mathcal M}(X - m+k)$ is constant on $\mathbb{Z}$. Using that $\sum_{s \in \mathbb{Z}} \theta_{\mathcal M}(s)^2 = 1$, we obtain that
	\begin{align*}
		\sum_{m \in \mathbb{Z}} \mathrm{Tr}(H_N \Gamma_m) 
			&= \mathrm{Tr}(H_N\Gamma) + \sum_{\vert k \vert \leq 2} \delta_k \mathrm{Tr}(\mathcal{H}^k\Gamma) \\
			&= \mathrm{Tr}(H_N\Gamma) + \delta_1 \mathrm{Tr}(d_{1} \Gamma) + \delta_2 \mathrm{Tr}( d_{2} \Gamma),
	\end{align*}
	with
\begin{equation*}
		\delta_k = \sum_{m \in \mathbb{Z}} \big (  \theta_{\mathcal M} (m)  \theta_{\mathcal M}(m+k) -  \theta_{\mathcal M}(m)^2  \big ) = - \frac{1}{2} \sum_{m\in \mathbb{Z}} \big ( \theta_{\mathcal M}(m) -  \theta_{\mathcal M}(m+k) \big )^2.
	\end{equation*}
	It remains to prove that $\vert \delta_k \vert \leq C \mathcal M^{-2}$. This follows from the fact that $\theta$ is smooth and has support in $[-1/4,1/4]$, so that we can restrict the sum to $ m  \in \big[- \frac{ \mathcal M}{2} , \frac{ \mathcal M}{2} \big]$, and $c_\mathcal M > C^{-1} \mathcal M$.
\end{proof}

To estimate the error in \eqref{eq:large_matrices_decomposn+}, we need the following bounds on $d_{1}$ and $d_{2}$.
\begin{lemma}\label{lem:d1d2estimate}
	There exists a universal constant $C > 0$ such that, for any trace-class operator $\Gamma$ with $\Gamma \geq 0$ and $\rm{Tr}\, \Gamma =1$ we have
	\begin{align}\label{eq:estimateMd1d2}
		\vert \mathrm{Tr} ( d_{1} \Gamma) \vert +
		\vert \mathrm{Tr} ( d_{2} \Gamma) \vert 
		\leq C \tr \Big( \sum_{i\neq j} v(x_i-x_j) \Gamma \Big) + C \|v\|_1K_H^3 \ell^{-3}  N \tr \big(  n_+ \Gamma \big) .
	\end{align}
\end{lemma}
\begin{proof}
	We start by noting the following bound. For $x,y \in \Lambda$, denoting $v_y(x) = v(x-y)$, we have 
	\begin{align}\label{eq:boundQnonhigh-v}
	Q^L v_y Q^L 
		&= \sum_{0< |p|,|q| < K_H}  \Big(\int_{\Lambda} v_y u_p u_q \Big) \ket{u_p}\bra{u_q} \leq C  \frac{K_H^3}{\ell^3} \lVert v\rVert_1 n_+,
	\end{align}
	where $\{u_p\}$ is the Neumann basis defined in (\ref{def:neubasis}). Therefore, expanding all the prefactors in (\ref{eq:d1}), we can bound $d_k$, $k \in \{1,2\}$, using the Cauchy-Schwarz inequality by
	\begin{align*}
	\pm d_k 
		&\leq C \sum_{i\neq j} v(x_i-x_j) + Q^L_i v(x_i-x_j) Q^L_i + Q^L_i Q^L_j v(x_i-x_j) Q^L_i Q^L_j \\
		&\leq C \sum_{i\neq j} v(x_i-x_j) + C K_{H}^{3} \frac{\|v\|_{L^1} }{\ell^3} n_+ N .
	\end{align*}
	From which the claim follows.
\end{proof}

Now we can combine Lemmas~\ref{lem:localization_largeMat} and~\ref{lem:d1d2estimate} to prove Proposition~\ref{thm:excitationrestriction}.

\begin{proof}[Proof of Proposition~\ref{thm:excitationrestriction}]
Let $\Gamma_0 = e^{-H_N/T} / \tr \big( e^{-H_N/T} \big)$ be the Gibbs state. Lemma \ref{lem:localization_largeMat} gives
\begin{equation}\label{eq:Ham-Q4-locM}
\mathrm{Tr}({H}_N \Gamma_0) \geq \mathrm{Tr}({H}_N \Gamma_{\leq})+ \mathrm{Tr}({H}_N \Gamma_{>}) - \frac{C}{\mathcal{M}^2}(\vert \mathrm{Tr}(d_1\Gamma_0)\vert +\vert \mathrm{Tr}(d_2\Gamma_0)\vert),
\end{equation}
where
\begin{equation}
\Gamma_{\leq} := \sum_{\lvert m  \rvert \leq \mathcal{M}/8} \Gamma_m, \qquad \Gamma_{>} := \sum_{\lvert m  \rvert > \mathcal{M}/8} \Gamma_m.
\end{equation}
Denoting by $\alpha_{\Gamma}:= \mathrm{Tr}(\Gamma_>),$ we have $1-\alpha_{\Gamma} = \mathrm{Tr}(\Gamma_{\leq}),$ and from the sub-additivity of the entropy $S(\Gamma) = - \tr \, \big( \Gamma \ln \Gamma \big)$ we have
\begin{equation}\label{eq:subadd-entropy}
S(\Gamma_0) \leq S(\Gamma_{\leq}) + S(\Gamma_{>})  = \alpha_{\Gamma} S\Big( \frac{\Gamma_{>}}{\alpha_{\Gamma}}\Big) + (1-\alpha_{\Gamma}) S\Big(\frac{\Gamma_{\leq}}{1-\alpha_{\Gamma}}\Big) - \alpha_{\Gamma} \log(\alpha_{\Gamma}) - (1-\alpha_{\Gamma}) \log(1-\alpha_{\Gamma}).
\end{equation}
Combining \eqref{eq:Ham-Q4-locM} and \eqref{eq:subadd-entropy} we obtain
\begin{align*}
F(\ell,N) &= \mathrm{Tr}(H_N \Gamma_0) -T S(\Gamma_0) \\
&\geq (1-\alpha_{\Gamma})\left[ \mathrm{Tr}\Big(H_N \frac{\Gamma_{\leq}}{1-\alpha_{\Gamma}}\Big)- T S \Big( \frac{\Gamma_{\leq}}{1-\alpha_{\Gamma}}\Big)\right] + \alpha_{\Gamma}\left[ \mathrm{Tr}\Big(H_N \frac{\Gamma_{>}}{\alpha_{\Gamma}}\Big)- T S\Big(\frac{\Gamma_>}{\alpha_{\Gamma}}\Big)\right] \\
&\quad  - \frac{C}{\mathcal{M}^2}(\vert \mathrm{Tr}(d_1\Gamma_0)\vert +\vert \mathrm{Tr}(d_2\Gamma_0)\vert) + T\alpha_{\Gamma} \log(\alpha_{\Gamma}) + T(1-\alpha_{\Gamma}) \log(1-\alpha_{\Gamma}).
\end{align*}
By the variational principle, the second term above is bigger than $\alpha_{\Gamma} F(\ell,N)$, and subtracting this quantity on both sides and dividing by $1-\alpha_{\Gamma}$, we obtain
\begin{align*}
F(\ell,N) &\geq \mathrm{Tr}\Big({H}_N \frac{\Gamma_{\leq}}{1-\alpha_{\Gamma}}\Big)- T S \Big( \frac{\Gamma_{\leq}}{1-\alpha_{\Gamma}}\Big) - \frac{C}{(1-\alpha_{\Gamma})\mathcal{M}^2}(\vert \mathrm{Tr}(d_1\Gamma_0)\vert +\vert \mathrm{Tr}(d_2\Gamma_0)\vert)\\
&\quad  + T\frac{\alpha_{\Gamma}}{1-\alpha_{\Gamma}} \log(\alpha_{\Gamma}) + T \log(1-\alpha_{\Gamma}).
\end{align*}
Lastly we see that
\[\alpha_{\Gamma}= \mathrm{Tr}(\Gamma_{>})\leq \mathrm{Tr}(\one_{\{n_+\geq \frac{\mathcal{M}}{8}\}}\Gamma)\leq \mathrm{Tr}\Big(8\frac{n_+}{\mathcal{M}}\Gamma\Big).\]
Defining $\Gamma_{\mathcal{M}}=\frac{\Gamma_{\leq}}{1-\alpha_{\Gamma}}$ and using the bounds in Lemma~\ref{lem:d1d2estimate}, we obtain the claim.
\end{proof}

%%%%%%%%%%%%%%%%%%%%%%%%%%%%%%%%%%%%%%%%%%%%%%%%%%
%%%%%%%%%%%%%%%%%SYMMETRISATION%%%%%%%%%%%%%%%%%%%%%%%%%

\section[Symmetrization]{Symmetrization: Proof of Theorem \ref{thm.sym}}\label{sec.sym}

In this section we prove Theorem \ref{thm.sym}, which estimates the error made replacing the Hamiltonian $H_N^{\rm{mod}}$ in \eqref{eq.Hmod} by its symmetrized version $H_N^{\rm{sym}}$. We begin by gathering some important properties of the symmetrization. Recall the definition \eqref{def:pz} of the symmetry $p_z$, for $z \in \mathbb Z^3$, and the symmetrized version $f^{\rm{s}}$ of a function $f$ in \eqref{def:gs}.

\begin{lemma}\label{lem:pz}
For all $z \in \mathbb Z^3$, $x$, $y \in \Lambda$ and $f:\; \mathbb{R}^3\to \mathbb{R}$ radial such that with $\mathrm{supp}(f) \subset B(0,R)$ for some $R \leq \ell /2$, we have
\begin{enumerate}
\item $f(p_z(x)-y)=0$ if one of the $\vert z_i\vert\geq 2$,
\item $f^{\rm{s}}(x,y)=f^{\rm{s}}(y,x)$,
\item  $\vert p_z(x)-y\vert \geq \vert x-y\vert$,
\item $u_p(p_z(x)) = u_p(x), \quad \text{for } p \in \frac{\pi}{\ell} \mathbb N_0^3$.
\end{enumerate}
\end{lemma}

\begin{proof}
Property 1 is a consequence of the assumption on the support of $f$. Property 2 follows from that $f$ is radial and $\vert p_z(x)-y\vert = \vert p_z(y)-x\vert$. 
Property 3 is obvious from the definition of $p_z$, observing that $x$, $y \in \Lambda$.
The last property follows from that the definition \eqref{def:neubasis} of the Neumann eigenbasis $(u_p)$, noting that $((-1)^{z+1}+1)/2 + z \in 2 \mathbb{Z}$ for all $z \in \mathbb{Z}$.
\end{proof}

We now proceed to the proof of Theorem \ref{thm.sym}. From (\ref{eq.Hmod}) and (\ref{def:hsym}), we need to find a lower bound on
\begin{align*}
H_N^{\mathrm{mod}} - H_N^{\mathrm{sym}} = \sum_{j=0}^3 (Q_j^{\rm{ren}} - Q_j^{\rm{sym}}) + Q_4^{\rm{ren}},
\end{align*}
where the $Q_j^{\rm{ren}}$ were defined in (\ref{eq:SF_DefQ0})-(\ref{eq:SF_DefQ4}) and the $Q_j^{\rm{sym}}$ were defined in (\ref{eq:Q_0sym})-(\ref{eq:Q_3sym}).

\paragraph{Symmetrization of $Q_{0}^{\rm{ren}}$.}
From $P= \frac{1}{\ell^3} | 1 \rangle\langle 1 |\leq \one$, we obtain
\begin{align*}
0\leq Q_0^{\rm{sym}} -  Q^{\rm{ren}}_0  &=  \frac 12 \sum_{i\neq j} P_i P_j ( g^{\rm{s}} - g + (g\omega)^{\rm{s}} - g\omega)  P_j P_i \\
&=  \frac 12 \sum_{z\neq 0} \sum_{i \neq j} P_i P_j \big( g(p_z(x_i) -x_j) + (g\omega)(p_z(x_i)-x_j) \big)  P_j P_i \\
&\leq \frac{CN^2}{|\Lambda|^2}  \sum_{z\neq 0} \int_{\Lambda^2} g(p_z(x)-y) \dd x \dd y,
\end{align*}
where we first used the definition of $g^{\rm{s}}$ in \eqref{def:gs} and that $g\omega \leq g$. Since $\supp g \subset \{|x|\leq R\}$, we can restrict the sum to $|z|\leq 1$ and the integration over $x$ to $\Lambda \setminus \Lambda_R$, with $\Lambda_R = [R,\ell-R]^3$. We obtain that
\begin{align}
 Q_0^{\rm{sym}} -  Q^{\rm{ren}}_0 
 	&\leq \frac{CN^2}{|\Lambda|^2}  \sum_{|z| \leq 1} \int_{\Lambda^2} \mathbf{1}_{\{\Lambda \setminus \Lambda_R\}}(x) g(p_z(x)-y) \dd x \dd y \leq C N \rho \widehat g(0) \frac R \ell. \label{eq.symQ0}
\end{align}

\paragraph{Symmetrization of $Q_{2}^{\rm{ren}}$.}
Recall the definition of $Q_2^{\rm{ren}}$ in (\ref{eq:SF_DefQ2}). We have three different kinds of terms from $Q_{2}^{\rm{ren}}-Q_{2}^{\rm{ren}}$ to treat, which we denote in a self-explanatory way by $PQPQ$, $PQQP$ and $PPQQ$. We start with
\begin{align*}
PQQP:=\sum_{i \neq j} P_i Q_j (g^{\rm{s}} - g+(g\omega)^{\rm{s}} - g\omega) Q_j P_i 
&\leq  \frac{2}{|\Lambda|} \sum_{z \neq 0} \sum_{i \neq j} Q_j \int_{\Lambda} g(p_z(x_j) - y) \dd y Q_j P_i \\
&\leq C \frac{N}{|\Lambda|} \widehat g(0) \sum_{j=1}^N Q_j \mathbf{1}_{\{\Lambda \setminus \Lambda_R\}}(x_j) Q_j.
\end{align*}
where we used the same properties as above plus the symmetry $\vert p_z(x)-y\vert = \vert p_z(y)-x\vert$.
To estimate this last sum, we first isolate the high momenta using $Q = Q^H + Q^L$ defined in \eqref{def:QL} to get
\begin{equation} \label{eq.sym01}
PQQP \leq C \rho \widehat g(0) \sum_{j=1}^NQ^L_{j} \mathbf{1}_{\{\Lambda \setminus \Lambda_R\}}(x_j) Q^L_{j} + C \rho \widehat g(0) n_+^H.
\end{equation}
For the low momenta part, we decompose with respect to the basis $(u_p)$ introduced in \eqref{def:neubasis}. Using the Cauchy-Schwarz inequality, we obtain
\begin{align*}
Q^L \mathbf{1}_{\{\Lambda \setminus \Lambda_R\}} Q^L &= \sum_{p,q \in \mathcal{P}_L} \Big(\int_{\Lambda \setminus \Lambda_R} \overline{u_p} u_{q}\Big) \ket{u_{p}}\bra{u_{q}} \leq C \frac{R}{\ell} |\mathcal{P}_L | Q^L = C \frac R \ell K_H^3 Q^L
\end{align*}
where we used $| u_p(x) u_{p'}(x) | \leq C |\Lambda|^{-1}$ and $| \Lambda \setminus \Lambda_R | \leq C R \ell^2$. Inserting the above in \eqref{eq.sym01}, we deduce that
\begin{equation}\label{eq.symQ2}
PQQP \leq C \rho \widehat g(0) \frac R \ell K_H^3 n_+^L + C \rho \widehat g(0) n_+^H.
\end{equation}
The next term to bound in $Q_{2}^{\rm{ren}}$ is the $PQPQ$.
By the Cauchy-Schwarz inequality $P_iQ_jP_jQ_i\leq P_iQ_jQ_jP_i +  Q_iP_jP_jQ_i$ and the symmetry of $g^{\rm{s}}$ and $(g\omega)^{\rm{s}}$  we obtain the same bound as \eqref{eq.symQ2}. 

In order to bound $QQPP$, we need to use $Q_4^{\rm{ren}}$. To this end we reconstruct the projector $\Pi_{ij}= Q_j Q_i + \omega(x_i-x_j) \left(P_j P_i + P_j Q_i + Q_j P_i\right)$ appearing in the definition of $Q_4^{\rm{ren}}$. We have
\begin{align*}
QQPP&:= \sum_{i \neq j} Q_i Q_j (g^{\rm{s}} - g) P_j P_i + \hc\\
&=  \Big(\sum_{i \neq j} \Pi_{ij}(g^{\rm{s}} - g) P_j P_i + \hc \Big)- \Big(\sum_{i \neq j} (Q_i P_j + P_i Q_j + P_i P_j) \omega(g^{\rm{s}} - g) P_j P_i + \hc\Big)
\end{align*}
We can bound the second term by $ Q_0^{\rm{sym}} -  Q^{\rm{ren}}_0 $ and $PQQP$, which have already been estimated, using the Cauchy-Schwarz inequalities and $\omega \leq 1$. The first term is also bounded using the Cauchy-Schwarz inequality, giving
\begin{align*}
QQPP &\leq  C\varepsilon^{-1} N \rho \widehat g(0) \frac R \ell + C \rho \widehat g(0) \frac R \ell K_H^3 n_+^L + C \rho \widehat g(0) n_+^H + \varepsilon \sum_{i \neq j}\Pi_{ij}^* \omega (g^{\rm{s}} - g) \Pi_{ij},
\end{align*}
for all $\varepsilon>0$.
The last term is estimated using that $Q_4^{\rm ren}$ and $\omega \leq 1$,
\begin{equation}\label{prop:gv}
g(p_z(x)-y)\leq C_0 v(x-y)
\end{equation}
for some $C_0>0$, which follows from that $g$ and $v$ satisfy \eqref{eq.propv2}, and $\vert p_z(x)-y\vert \geq \vert x-y\vert$ from Lemma \ref{lem:pz}.
Choosing $\varepsilon = C_0/100$, we conclude that
\begin{equation}\label{err:Q2}
  Q_2^{\rm{sym}} - Q^{\rm{ren}}_2 \leq  C \rho \widehat g(0) \frac R \ell K_H^3 n_+^L + C \rho \widehat g(0) n_+^H+ C N \rho \widehat g(0) \frac R \ell+\frac{1}{100}{Q}_4^{\rm ren}
\end{equation}
Note that the first two errors can be estimated by a fraction of the gap operator $G$ defined in \eqref{eq:gap}, as soon as $K_H \geq C K_\ell^2$ and $K_\ell K_H^3 \leq (\rho a^3)^{- \frac 1 2}$. The last fraction of $Q_4^{\rm{ren}}$ is absorbed by the positive $\frac 1 2 Q_4^{\rm{ren}}$ in $H_{N}^{\rm{mod}}$.

\paragraph{Symmetrization of $Q_{1}^{\rm{ren}}$.}
From the Cauchy-Schwarz inequality, we have
\begin{align}\label{eq:Q_1sym_diff}
Q_{1}^{\rm{sym}} - Q_{1}^{\rm{ren}} \leq C( QPPQ + Q_0^{\rm{sym}} -  Q^{\rm{ren}}_0) \leq C \rho a(  \frac R \ell N +  \frac R \ell K_H^3 n_+^L + n_+^H). 
\end{align} 
which we already estimated in \eqref{eq.symQ0} and \eqref{eq.symQ2}, and that are absorbed by a fraction of the gap operator $G$ for $\rho a^3$ small enough. 
It remains to note that $Q_1^{\rm{sym}} = 0$, indeed
\begin{align*}
Q_1^{\rm{sym}} &= \sum_{z} \sum_{i \neq j} P_i P_j g(p_z(x_i)-x_j) P_i Q_j + \hc = \frac{1}{|\Lambda|} \sum_z \sum_{i\neq j} P_j \int_{\Lambda} g(p_z(x) - x_j ) \dd x  P_i Q_j + \hc \\
&= \frac{1}{|\Lambda|}\sum_{i\neq j} P_j \int_{\mathbb R^3} g(x - x_j ) \dd x  P_i Q_j + \hc = \frac{\widehat{g}(0)}{\vert \Lambda\vert } \sum_{i \neq j } P_i P_j Q_j = 0.
\end{align*}

\paragraph{Symmetrization of $Q_{3}^{\rm{ren}}$.}
In view of Lemma \ref{lem.Q3loc1}, it is enough to estimate
\begin{align*}
 Q_{3,L}^{\rm{sym}}-Q_{3,L}=\sum_{z\neq 0} \sum_{i \neq j} (P_i Q^L_{j} g(p_z(x_i) -x_j)Q_i Q_j + \hc).
\end{align*}
For the rest of the proof we use the notation $g_{\neq 0}:=\sum_{z\neq 0} g(p_z(x_i) -x_j)$.
We want to reconstruct $ Q_4^{\rm{ren}}$ as before, we write
\begin{align}
 Q_{3,L}^{\rm{sym}}-Q_{3,L} &=  \Big( \sum_{i \neq j} P_i Q^L_{j} g_{\neq 0} \Pi_{ij} + \hc \Big) - \Big(\sum_{i \neq j} P_i Q^L_{j} g_{\neq 0}\omega (P_iP_j +P_iQ_j + Q_i P_j) + \hc\Big)\label{eq:Q'es2}
\end{align}
From the Cauchy-Schwarz inequality we obtain, as before using \eqref{prop:gv}, that
\begin{align*}
 Q_{3,L}^{\rm{sym}}-Q_{3,L} \leq C( QPPQ + Q_0^{\rm{sym}} -  Q^{\rm{ren}}_0 ) + \frac{1}{100}  Q_4^{\rm{ren}},
\end{align*}
where the first term is estimated in (\ref{eq:Q_1sym_diff}). This concludes the proof of Theorem \ref{thm.sym}.
\qed

%%%%%%%%%%%%%%%%%%%%%%%%%%%%%%%%%%%%%%%%%%%%%%%%%%
%%%%%%%%%%%%%%%%%%%CNUMBER%%SECOND Q%%%%%%%%%%%%%%%%%%%%%%%%%

\section{C-number substitution}\label{sec:cnum}

In this section we prove Lemma \ref{lem.secondquant}, Theorem \ref{thm.cnumber} and Lemma \ref{lem.largez}. 

%%%%%%%%%%%%%%%%%%%%%%%%%%%%%%%%%%%%%%%%%%%%%%%%%%%%%
%%%%%%%%%%%%%%%%%%%%SECOND Q%%%%%%%%%%%%%%%%%%%%%%%%%%%%%

\subsection{Proof of Lemma~\ref{lem.secondquant}}
The second quantization of the kinetic energy is immediate. Then by definition of $n_+$ and from Lemma~\ref{lem:sq} we have $$Q_0^{\rm{sym}} = \frac{\widehat{g}(0)+\widehat{g\omega}(0)}{2|\Lambda|}  (N-n_+)(N-n_+-1).$$
Again using Lemma~\ref{lem:sq}, we have 
\begin{align*}
Q_2^{\rm{sym}} = \frac{1}{|\Lambda|}\sum_{p \in {\Lambda}^*_+} \left(\widehat{g(1+\omega)}(p)+\widehat{g(1+\omega)}(0)\right) a_0^* a_p^* a_p a_0 + \frac{1}{2}\widehat{g}(p)\big( a_0^* a_0^* a_p a_p + \hc \big)
\end{align*}
From that $(N-n_+)(N-n_+-1)= N(N-1) - 2n_+ (N-n_+) - n_+(n_+-1) $ and that $\sum_{p \in {\Lambda}^*_+} a_0^* a_p^* a_p a_0 = n_+ (N-n_+)$ we recover the first two lines of (\ref{eq:lem.secondquant}).
The second quantization of the $Q_{3,L}^{\rm{sym}}$ is not expressed as simply. We start from
\begin{equation}\label{eq. second quant Q3}
	Q_{3,L}^{\rm{sym}}=\sum_{q,p,k,s\in\Lambda^*}\braket{u_{q} \otimes u_{p}|P_xQ_y^Lg^s(x,y)Q_xQ_y\vert  u_{k}\otimes u_{s}}a^*_{q}a^*_{p}a_{k}a_{s}+\hc
\end{equation}
where we observe that we can restrict the sum to $q=0$, $p\in \mathcal P_L$ and $k,s \in \Lambda^*_+$. We are left to compute the integral
\[\braket{u_{0} \otimes u_{p}|P_xQ_y^Lg^s(x,y)Q_xQ_y\vert  u_{k}\otimes u_{s}}=\frac{1}{|\Lambda|^{1/2}}\int_{\Lambda^2}u_p(y)g^s(x,y)u_k(x)u_s(y)dxdy.\]
Using the trigonometric formulas we obtain
\begin{align*}
u_p(y)u_s(y)&=\frac{1}{2^3 \vert \Lambda\vert  }\prod_{i=1}^3c_{p_i}c_{s_i}(\cos((p_i-s_i)y_i)+\cos((p_i+s_i)y_i))\\&=\frac{1}{2^3|\Lambda|^{1/2}}\prod_{i=1}^3c_{p_i}c_{s_i}\sum_{ \sigma_1,\sigma_2,\sigma_3 \in \{\pm\}}\frac{u_{p_1 + \sigma_1  s_1,p_2 + \sigma s_2,p_3 + \sigma_3 s_3}}{c_{p_1 +\sigma_1  s_1}c_{p_2 + \sigma_2 s_2}c_{p_3 +  \sigma_3s_3}},
\end{align*}
where $c_k$ was defined in (\ref{def:neubasis}). Using Lemma~\ref{lem:sq} we obtain
\begin{align*}
&\int_{\Lambda^2}u_p(y)g^s(x,y)u_k(x)u_s(y)dxdy\\ &\qquad  =\frac{1}{2^3|\Lambda|^{1/2}}\prod_{i=1}^3c_{p_i}c_{s_i}\sum_{ \sigma_1,\sigma_2,\sigma_3 \in \{\pm\}}\frac{\delta_{k_1,|p_1 + \sigma_1 s_1|}\delta_{k_2,|p_2 + \sigma_2 s_2|}\delta_{k_3,|p_3 +\sigma_3 s_3|}}{c_{|p_1 +\sigma_1  s_1|}c_{|p_2 \sigma_2 s_2|}c_{|p_3 +\sigma_3 s_3|}}\widehat{g}(k).
\end{align*}
Noting that we can replace $\sum_{p \in \mathcal P_L} \sum_{ \{\pm\}}$ by $\sum_{p\in \mathcal P_L^{\mathbb{Z}}}$, using the convention $a_p= a(u_p) = a_{(\vert p_1\vert, \vert p_2\vert,\vert p_3\vert)}$ and changing variable $s = k-p$ we obtain
\begin{equation}
	Q_{3,L}^{\rm{sym}}=\frac{1}{\vert\Lambda\vert}\sum_{p\in P_L^{\mathbb{Z}},k\in\Lambda^*_+ \setminus\{p\}}c(p,k)\widehat{g}(k)a^*_{0}a^*_{p}a_{k}a_{p-k}+\hc
\end{equation}
where $c(p,k)$ is given by
\begin{equation}\label{def. of cqk}
	c(p,k)=\prod_{i=1}^3\frac{c_{p_i}c_{k_i-p_i}}{2c_{k_i}}2^{\delta_{p_i,0}} = \prod_{i=1}^3 \frac{c_{k_i-p_i}}{c_{p_i} c_{k_i}}.
\end{equation}

\qed

%%%%%%%%%%%%%%%%%%%%%%%%%%%%%%%%%%%%%%%%%%%%%%%%%%%%%%
%%%%%%%%%%%%%%%%%%%%%CN%%%%%%%%%%%%%%%%%%%%%%%%%%%%%%%%

\subsection{Proof of Theorem \ref{thm.cnumber}}

The $c$-number substitution is an important step of Bogoliubov's approach to the dilute Bose gas and it is well known that it can be performed in a rigorous way.
We start from Lemma~\ref{lem.secondquant}, where the Hamiltonian $H_N^{\rm{sym}}$  has been written in second quantization (\ref{eq:lem.secondquant}). It is initially defined on $L^2(\Lambda)^N$ but naturally extends to the Fock space $\mathscr{F}(L^2(\Lambda))$.

We subtract the leading order of the energy and introduce a chemical potential $\mu$, which we will choose later, in order to finely tune the number of particles. Thus, we define
\begin{align}
\mathcal H_{\mu}^{\rm{sym}}&:= \frac{\widehat{g\omega}(0)}{2|\Lambda |}   a_0^*a_0^* a_0a_0  + \sum_{p \in {\Lambda}^*_+} \Big( \tau(p) a_p^* a_p + \frac{\widehat{g}(p)}{|\Lambda |} a_0^* a_p^* a_p a_0 + \frac{ \widehat{g}(p) }{2 |\Lambda |}\big( a_0^* a_0^* a_p a_p + \hc \big) \Big) \nonumber \\
& \quad + \frac{1}{|\Lambda|} \sum_{k \in {\Lambda}^*_+, \, p \in \mathcal P_L^{\mathbb Z} }c(p,k) \widehat{g}(k) \big( a_0^* a_p^* a_{p-k} a_k + \hc \big) + \frac{1}{|\Lambda|} \sum_{p\in {\Lambda}^*_+}(\widehat{g \omega}(0)+\widehat{g \omega}(p))a_0^* a_p^*a_p a_0 \nonumber \\
	&\quad - \mu \mathcal N + \rho a \frac{\mathcal{N}^m}{N^m}+  \frac{8\pi a(\rho a^3)^\frac{1}{4}}{\vert\Lambda\vert}(n_0^2 -N^2) \nonumber \\
	&= H_N^{\rm{sym}}-\frac{\widehat{g}(0)(N(N-1)-n_+(n_+-1))}{2|\Lambda |} - \mu \mathcal{N} + \rho a \frac{\mathcal{N}^m}{N^m} + \frac{8 \pi a(\rho a^3)^\frac{1}{4}}{\vert\Lambda\vert}(n_0^2 -N^2) \label{eq:def_H_mu}
\end{align}
where the last term, which is negative on $L^2(\Lambda)^N$, is added to ensure convexity in ``$n_0$" (see Remark~\ref{rem:convexity}) and we introduced the number operator which acts on the Fock space on each sector as  
\begin{equation}
(\mathcal{N} \Psi)^{(N)} = N \Psi^{(N)}, \qquad \text{for any } \Psi \in \mathscr{F}(L^2(\Lambda)).
\end{equation}

Note that $$\widehat{g}(0) n_+^2|\Lambda|^{-1} \leq C \rho a n_+^H+C a \ell^{-3} n_+^Ln_+ \leq C (\rho a \ell^2 K_{H}^{-1} + \mathcal M a \ell^{-1}) G$$ 
where $G$ is the gap operator defined in (\ref{eq:gap}). Using that, from assumptions, $C K_{\ell}^4\leq  K_H $ and $\mathcal M \leq C^{-1} \ell / a$, we obtain that for $\rho a^{3}$ small enough,
\begin{align}\label{eq. now Fockspace}
		-T\log {\rm{Tr}}_{L^2({\Lambda^N})}e^{-\frac{1}{T} (H^{{\rm{sym}}}_{N}+\frac{1}{2}G)}& \geq 4\pi a \frac{N^2}{|\Lambda|}-T\log {\rm{Tr}}_{\mathscr{F}(L^2(\Lambda))}e^{-\frac{1}{T} (\mathcal{H}_{\mu}^{{\rm{sym}}}+\frac{1}{4}G)} + \mu N- C \rho a.
	\end{align}
In the above we artificially inserted $\frac{\mathcal{N}^m}{N^m}$ (which is equal to $1$ on the $N$ particle sector) at the expanse of an error of order $\rho a$. This will ensure that we never have too many particles after the c-number substitution via \eqref{new gap term} and will be used to prove Lemma \ref{lem.largez}.

We follow \cite{LieSeiYng-05} and perform the c-number substitution. The decomposition $L^{2}(\Lambda)=\mathrm{Ran}P \oplus \mathrm{Ran}Q$ leads to the splitting of the bosonic Fock space $\mathscr{F}(L^{2}(\Lambda))=\mathscr{F}(\mathrm{Ran}P)\otimes\mathscr{F}(\mathrm{Ran}Q)$.
	Denoting by $\Omega$ the vacuum vector, we introduce the class of coherent states in $\mathscr F(\mathrm{Ran}P)$, labeled by $z \in \mathbb{C}$,
		\begin{equation}\label{eq:coherent_z}
		|z\rangle = e^{-\big(\frac{|z|^2}{2} + z a^{\dagger}_0\big)}\, \Omega,
	\end{equation}
	which are eigenvectors for the annihilation operator of the condensate. One easily checks that
	\begin{equation}\label{eq:coherent_prop}
		a_0 |z\rangle = z\, |z \rangle \;\;\;\;\text{and}\;\;\;\;	1 = \int_{\mathbb{C}} |z\rangle \langle z| \, \dd z
	\end{equation}
where $\dd z = \pi^{-1} \dd x\dd y$, $z=x+iy$.  For any $\Psi \in \mathscr{F}(L^2(\Lambda))$, we denote by $\Phi(z) = \langle z | \Psi \rangle$ the partial inner product, which is in $\mathscr{F}(\mathrm{Ran} Q)$.
We write any monomial in $a_0$ and $a^*_0$ as a polynomial sum of terms in anti-normal order thanks to the commutation rules. Using theses definitions and the inequality \cite[Eq.(7)]{LieSeiYng-05} we are allowed then to replace $a_0$ with a complex number (called upper symbol) according to the following substitution rules
\begin{equation}\label{eq. C-number substitution}
	\begin{aligned}
		&a_0\mapsto z,\qquad &a_0a_0\mapsto z^2,\qquad & a_0a^*_0 \mapsto |z|^2,\\
		&a_0^*\mapsto \overline{z},\qquad &a_0^*a_0^*\mapsto \overline{z}^2,\qquad &a_0a_0a^*_0 a^*_0 \mapsto |z|^4
	\end{aligned}
\end{equation}
which also give the substitutions
\begin{equation}
a_0^*a_0\mapsto \vert z\vert^2-1, \qquad a_0^*a_0^*a_0a_0\mapsto \vert z\vert^4-4\vert z\vert^2+2
\end{equation}
 before integrating over $z$ in front of the trace in \eqref{eq. now Fockspace}.
For the specific case of a power of the $n_0$, we use again the commutation rules to write it as polynomial of anti-normal ordered monomials in $a_0,a_0^*$ and apply the rules \eqref{eq. C-number substitution} to define the polynomial in $|z|^2$
\begin{equation}
n_0^m = (a^*_0a_0)^m \mapsto p_m(z) = |z|^{2m} + \text{smaller order terms},
\end{equation}
where the smaller order terms have explicit, constant coefficients. 

We first bound below the new term using that
\begin{equation}\label{new gap term}
\rho a\frac{\mathcal{N}^m}{N^m} \geq  \rho a\frac{n_0^m+n_0^{m-1}n_+ +n_0^{m-2}n_+^2}{N^m}.
\end{equation}
The upper symbol of the right hand side is
\begin{equation}\label{eq:polynm}
\frac{\rho a}{N^m}\big(p_m(z)+p_{m-1}(z)n_++p_{m-2}(z)n_+^2\big).
\end{equation} 
We observe that, for $|z|^2 > N,$ (and $N$ sufficiently large) the polynomials satisfy the bound $p_m(z) \geq \frac{1}{2} |z|^{2m}$ and therefore, in this region, \eqref{eq:polynm} can be bounded from below by 
\begin{equation}\label{whataniceequation}
\ \frac{\rho a}{2 N^m}\big(|z|^{2m}+|z|^{2m-2}n_+ + |z|^{2m-4} n_+^2\big).
\end{equation}
When $|z|^2 \leq N$, the error made replacing \eqref{eq:polynm} by \eqref{whataniceequation} is of order $N^{-1} \rho a$.
The last term in (\ref{eq:def_H_mu}) becomes
\begin{equation}
\frac{8\pi a(\rho a^3)^\frac{1}{4}}{\vert\Lambda\vert}(n_0^2 -N^2)\mapsto \frac{8 \pi a(\rho a^3)^\frac{1}{4}}{\vert\Lambda\vert}(\vert z\vert^4-3\vert z\vert^2+1- N^2).
\end{equation}
Using that, for $m >3$,
\begin{align}\label{eq:est_c_numb}
3 \frac{a}{\vert\Lambda\vert} \vert z\vert^2 - 2 \leq C \frac{a}{\vert\Lambda\vert}\left(K_{\ell} N \one_{|z|^2 \leq K_{\ell} N} +   \frac{|z|^{2m}}{N^m} N K_{\ell}^{1-m}\one_{|z|^2 > K_{\ell}N}\right) \leq C \rho a \left(K_{\ell} +  \frac{|z|^{2m}}{N^m}\right)
\end{align}
we recover the last term of \eqref{def.Hmu} up to an error $C \rho a$ (because from the assumptions $K_{\ell} \leq (\rho a^3)^{-1/26}$) and a fraction of the gap $\mathcal{G}(z)$.

When proceeding to the c-number substitution in the other terms of the Hamiltonian $\mathcal{H}^{\mathrm{sym}}_{\mu}$ we obtain errors from the terms proportional to $a^*_0 a_0$ and $a^*_0a^*_0 a_0 a_0$, which are 
\begin{equation}
\mathcal{E}= \widehat{g \omega }(0)\frac{2 -4 \vert z\vert^2}{2\vert \Lambda\vert}-\sum_{p \in {\Lambda}^*_+} \frac{\widehat{g}(p)+\widehat{g \omega}(0)+\widehat{g \omega}(p)}{\vert \Lambda\vert} a^*_pa_p - \mu (n_+-1).
\end{equation}
Note that $|\widehat{g \omega}(p)|+ |\widehat{g}(p)| \leq Ca$. Therefore, we can estimate the first term above as in (\ref{eq:est_c_numb}). The second and third terms are estimated by by $ C (a |\Lambda|^{-1} + \mu) n_+  \leq \ell^{-2} n_+$ which is absorbed by a fraction of the gap (recall that $0\leq 10\mu\leq \ell^{-2}$ by assumption). In the end, we recover $\mathcal{H}_{\mu}(z)$ as well as $\mathcal{G}(z)$ and the inequality stated in Theorem \ref{thm.cnumber}.
\qed

\subsection{Proof of Lemma \ref{lem.largez}}
We conclude this section with the proof of Lemma \ref{lem.largez}, which gives a lower bound for the Hamiltonian in the physically irrelevant region $|z|^2 \geq K_{\ell}^{1/4} N$. Starting from the Hamiltonian \eqref{def.Hmu}, we use that
	\begin{align*}
	\frac{\vert z\vert^2}{|\Lambda|}  \sum_{p\in {\Lambda}^*_+} \left(\widehat g(p) + \widehat{g\omega}(p)\right)   a_p^* a_p \leq  C \rho a |z|^2 \frac{n_+}{N}
	\end{align*}
	 and we drop all the non-negative terms except the kinetic energy and the spectral gaps. Using that $\vert z\vert^2\geq K_{\ell}^{1/4} N$, we obtain
	\begin{align}\label{eq. first lower bound for large z}
		\mathcal{H}_{\mu}(z)+\mathcal{G}(z)&\geq \sum_{p\in \Lambda^*}\big(\tau (p) a_p^*a_p+\frac{\widehat{g}(p)}{2\vert \Lambda\vert}(z^2a_p^*a_p^*+h.c)\big)+Q_{3,L}^{\rm{sym}}(z)-\mu\vert z\vert^2 \nonumber\\ &\qquad +\frac{\rho a }{2N^m}(\vert z\vert^{2m}+\vert z\vert^{2m-2}n_+ +\vert z\vert^{2m-4}n_+^2). 
	\end{align}
	We now show that a fraction of this last term absorbs all the negative terms of \eqref{eq. first lower bound for large z}. First note that $\mu \leq C \ell^{-2} \leq K_{\ell}^{-1} \rho a$. To bound $Q_{3,L}^{\rm{sym}}$ we use the Cauchy-Schwarz inequality, we obtain for all $\delta>0$
	\begin{align*}
	\pm Q_{3,L}^{\rm{sym}} 
		&\leq \delta \sum_{p\in P_L^{\mathbb{Z}},k\in\Lambda^*\cup\{0\}} k^2 a^*_ka_{k} + \delta^{-1} |z|^2 \frac{1}{\vert\Lambda\vert^2}\sum_{p\in P_L^{\mathbb{Z}},k\in\Lambda^*_+} \frac{\widehat{g}(k)^2}{k^2} a^*_{p}a_{p-k} a^*_{p-k} a_p \\
		&\leq  C \delta |\mathcal P_L| \sum_{p\in \Lambda^*} \tau_p a_p^*a_p + \delta^{-1} \frac{C|z|^2}{\ell^3}(a n_+ + a^2 \ell^{-1} n_+^2)\\
		&\leq \frac{1}{2} \sum_{p\in \Lambda^*} \tau_p a_p^*a_p +  C \rho a K_H^3K_\ell^{-2} a\ell^{-1} |z|^2 ( n_+ + a\ell^{-1} n_+^2)
	\end{align*}
	where we used (\ref{eq.gomega0app0}) to estimate the sum and chose $\delta = \frac{1}{2}C^{-1}  |\mathcal P_L|^{-1}$. The second term above is bounded by a fraction of the last term in (\ref{eq. first lower bound for large z}) using the assumptions on $m$. Recalling that $\tau_p$ is given by (\ref{eq:tau}), we can replace it by $p^2$ by absorbing the negative $n_+\ell^{-2}$ and $K_H n_+^H\ell^{-2}$ terms, and we are left with
	\begin{equation}\label{eq. bounding Q2 in large z}
		\sum_{p\in \Lambda^*}\big(\frac{1}{2}p^2 a_p^*a_p+\frac{\widehat{g}(p)}{2\vert \Lambda\vert}(z^2a_p^*a_p^*+h.c)\big)=\sum_{p\in \Lambda^*}\frac{1}{2}p^2 d_p^*d_p-\frac{\vert z\vert^4}{\vert \Lambda\vert^2}\sum_{p\in \Lambda^*}\frac{\widehat{g}(p)^2}{2p^2}a_pa_p^*
	\end{equation}
where
\[d_p=\frac{z^2\widehat{g}(p)}{\vert \Lambda\vert p^2}a_p^*+a_p.\]
Then the first term of \eqref{eq. bounding Q2 in large z} is positive and dropped for a lower bound and the second is bounded by a fraction of the last term of \eqref{eq. first lower bound for large z} if $m > 2\eta^{-1} + 14$, similarly as before.
\qed

%%%%%%%%%%%%%%%%%%%%%%%%%%%%%%%%%%%%%%%%%%%%%%%%%%
%%%%%%%%%%%%%%%%%%%%%%3QTERM%%%%%%%%%%%%%%%%%%%%%%%%%%

\section{Bounds on the 3Q terms}\label{sec:3Q}

In this section we prove the bounds on the 3Q terms stated in Lemma \ref{lem.Q3loc1} and Theorem \ref{thm.Q3}.
We start with Lemma \ref{lem.Q3loc1}, which estimates the error made by replacing $Q_3^{\mathrm{ren}}$ in \eqref{eq:SF_DefQ3} by $Q_{3,L}$ in \eqref{eq:defQ3low}.

\begin{proof}[Proof of Lemma \ref{lem.Q3loc1}]
From the definitions we have
\begin{equation}
 Q_3^{\text{ren}} -  Q_{3,L} = \sum_{i \neq j} (P_i Q^H_j g(x_i -x_j)Q_j Q_i + \hc ).
\end{equation}
In the right-hand side we aim at reconstructing $ Q_4^{\rm{ren}}$ as 
\begin{align}
 \sum_{i \neq j} (P_i Q^H_j g Q_j Q_i +\hc)  &=   \sum_{i \neq j} \big( P_i Q^H_j g \, \Pi_{ij} + \hc \big)
 -  \sum_{i \neq j} P_i Q^H_j g\omega (P_jP_i +P_jQ_i + Q_j P_i) + \hc \label{eq:Q'estim2}
\end{align}
We use a weighted Cauchy-Schwarz inequality on both terms. Using that $g \leq v$, the first term in \eqref{eq:Q'estim2} is controlled by
\begin{align*}
C\delta^{-1} \sum_{i \neq j} P_i Q^H_j g Q^H_j P_i + \delta Q_4^{\text{ren}}
&= C \delta^{-1}  \widehat g(0) \frac{n_0 n_+^H}{\vert \Lambda \vert} + \delta Q_4^{\text{ren}}
\end{align*}
for all $\delta >0$. The other terms can be estimated similarly to above. For instance,
\begin{align}
\sum_{i\neq j} ( P_i Q^H_j g \omega Q_j P_i  + \hc) &\leq
C \delta^{-1}\sum_{i\neq j} P_i Q^H_j g \omega Q^H_j P_i  +  \delta \sum_{i \neq j } P_i Q_j g \omega Q_j P_i \nonumber \\
&\leq C \widehat g(0) \frac{n_0}{\vert \Lambda \vert} \big( \delta^{-1} n_+^H +\delta  n_+\big),
\end{align}
where we used $\widehat{g\omega}(0) \leq \widehat{g}(0)$.
We collect the previous inequalities to obtain 
\begin{equation}
\vert \langle Q_3^{\rm{ren}} \rangle_\Psi  - \langle Q_{3,L} \rangle_\Psi \vert
 \leq \delta \langle Q_4^{\rm{ren}} \rangle_\Psi  + C \rho \widehat{g}(0) \big( \delta \langle n_+ \rangle_\Psi +\delta^{-1} \langle n_+^H \rangle_\Psi \big), 
\end{equation}
where we bounded $n_0 \leq N$. Choosing $\delta = \varepsilon  K_{\ell}^{-2}$, the two last terms are bounded by spectral gaps if $K_H \geq C K_\ell^4$.
\end{proof}

Then $Q_{3,L}$ was symmetrized in Theorem \ref{thm.sym}, leading after c-number substitution to $Q_{3,L}^{\mathrm{sym}}(z)$, defined in \eqref{def:Q3Lz}. The following Lemma shows that in $Q_{3,L}^{\mathrm{sym}}(z)$, only the soft pairs contribute, up to errors controlled by spectral gaps.

\begin{lemma}\label{lem.Q3loc2}
For all $\varepsilon >0$ there exists a $C>0$ such that the following holds. Let $v$ be a positive, radially symmetric potential with scattering length $a$ and assume $\rho a^3 \leq C^{-1}$. Then for all $|z|^2 \leq K_{\ell}^{1/4}N$ and $\mathcal M \leq C^{-1} \rho \ell^3 K_H^{-3} K_\ell^{-17/4}$ we have
	\begin{equation}
		   Q_{3,L}^{\rm{sym}}(z) -Q_3^{\rm{soft}}(z)  \geq-  \varepsilon \mathcal G(z).
	\end{equation}
where
\begin{equation}\label{Def. Q3soft}
	Q_{3}^{\rm{soft}}(z) = \frac{1}{\lvert \Lambda\rvert} \sum_{p \in \mathcal{P}_L^{\mathbb{Z}},k\in \mathcal P_H } c(p,k)\widehat{g}(k) \overline{z}a^*_p a_{p-k}a_k + \hc,
\end{equation}
and the spectral gaps $\mathcal G(z)$ are defined in \eqref{def.Gz}.
\end{lemma}

\begin{proof}
Recalling the definition of $  Q_{3,L}^{\rm{sym}}(z)$ in \eqref{def:Q3Lz} with $\mathcal{P}_L^{\mathbb{Z}}=\{ p\in \frac{ \pi}{\ell} \mathbb{Z}^3, \;\; 0<\vert p\vert \leq \frac{K_H}{\ell}\}$ we take the difference
	\begin{equation}
		 Q_{3,L}^{\rm{sym}}(z) - Q_3^{\rm{soft}}(z)
		= \frac{1}{\vert \Lambda \vert} \sum_{\substack{p \in \mathcal{P}_L^{\mathbb{Z}},k\in \mathcal P_L \\p \neq k}}  c(p,k) \widehat g(k) \big(  \overline{z} a_p^* a_{p-k} a_k + \hc \big) .
	\end{equation}
We use the Cauchy-Schwarz inequality with weight $ \delta >0$ and deduce
	\begin{align}
		 Q_{3,L}^{\rm{sym}}(z) - Q_3^{\rm{soft}}(z) \geq - C\frac{\widehat g(0)}{\vert\Lambda \vert} \sum_{\substack{p \in \mathcal{P}_L^{\mathbb{Z}},k\in \mathcal P_L\\ p \neq k}} \big(  \delta   \vert z\vert^2 a_p^* a_p   + \delta^{-1}  a_k^* a_{p-k}^* a_{p-k} a_k  \big), \label{eq:Q3soft1}
	\end{align}
where we bounded the $c(p,q)$ by a constant.

Recall that the commutation relation $[a_p,a_q^*]=\delta_{p,q}$ only holds when $p$ and $q$ are in $\pi\ell^{-1}\mathbb{N}^3_0$. To deal with the negative components that can appear due to the set where $p$ belongs, we use that $a_p=a_{(\vert p_1\vert,\vert p_2\vert,\vert p_3\vert)}$. In particular, note that
\begin{equation}
\sum_{p \in \mathcal{P}_L^{\mathbb Z}} a_p^* a_p \leq 8 \sum_{p \in \mathcal P_L} a_p^* a_p,
\end{equation}
and therefore the first term of \eqref{eq:Q3soft1} is bounded by $n_+$ and a cardinal of $\mathcal P_L$. Similarly in the second term, we bound the $p$-sum by $Cn_+$ and the $k$-sum by $n_+^L$. Choosing $\delta =\varepsilon C^{-1} K_\ell^{-9
/4} K_H^{-3}$ and using that $\rho_z \leq \rho K_{\ell}^{1/4
}$ we get
	\begin{align}
	Q_{3,L}^{\rm{sym}}(z) - Q_3^{\rm{soft}}(z)&\geq - C \delta \rho \widehat g(0) K_\ell^{\frac{1}{4}} K_H^3 n_+ - C \delta^{-1} \rho \widehat{g}(0) \frac{ n_+ n_+^L }{\rho \ell^3} \\
&\geq - \varepsilon \frac{n_+}{\ell^2}- C \frac{K_H^3 K_\ell^5}{\varepsilon \rho \ell^3} \frac{n_+ n_+^L}{ \ell^2}.
	\end{align}
These errors can be absorbed in spectral gaps if $\mathcal M \leq C^{-1} \rho \ell^3 K_H^{-3} K_\ell^{-17/4}$.
\end{proof}

In order to prove Theorem \ref{thm.Q3}, we need the following approximations of $\widehat{g\omega}(0)$.

\begin{lemma}\label{lem:gomega.approx}
The following estimates hold.
\begin{enumerate}
\item For all $\Psi \in \mathscr F^\perp$,
\begin{equation*}
\Big| \sum_{p \in \Lambda^*_+} \widehat{g\omega}(p) \langle a_p^* a_p \rangle_\Psi - \widehat{g\omega}(0) \langle n_+ \rangle_\Psi \Big| \leq C \widehat{g}(0) \langle n_+^H \rangle_\Psi + C \widehat{g}(0) K_H^2 R^2 \ell^{-2} \langle n_+ \rangle_\Psi.
\end{equation*}
\item Moreover,
\begin{align}
\Big\vert \widehat{g\omega}(0) - \frac{1}{8\vert \Lambda \vert} \sum_{\substack{k \in \frac{\pi}{\ell}\mathbb Z^3\setminus\{0\}}} \frac{\widehat g(k)^2}{2k^2} \Big\vert &\leq C a^2 \ell^{-1} \label{eq.gomega0app0} \\ 
\frac{1}{\vert \Lambda \vert}  \sum_{k \in \mathcal P_L^{\mathbb Z}} \frac{\widehat g(k)^2}{2k^2} &\leq C K_H a^2 \ell^{-1}, \label{eq.gomega0app}
\end{align}
\end{enumerate}
\end{lemma}

\begin{proof} 1. For the first estimate we split the sum into high and low momenta,
\begin{equation}
\sum_{p\in \Lambda^*_+}\widehat{g\omega}(p)a_p^*a_p=:\mathcal{S}(\mathcal{P}_L)+\mathcal{S}(\mathcal{P}_H).
\end{equation}
The high momenta part is controlled by $n_+^H$,
\begin{equation}
\mathcal{S}(\mathcal{P}_H)= \sum_{p\in \mathcal{P}_H}\widehat{g\omega}(p)a_p^*a_p\leq \widehat{g}(0)n_+^H,
\end{equation}
where we used $\widehat{g\omega}(0) \leq \widehat{g}(0)$. For low momenta, we use a Taylor approximation $\widehat{g \omega}(p) \simeq \widehat{g \omega}(0) + \mathcal{O}(p^2 \widehat{g}(0) R^2)$, which follows using $\widehat{g \omega}'(0) =0$ (by radiality) and $|\widehat{g \omega}''(p)| \leq \widehat{g}(0) R^2$ (since $\mathrm{supp}(g) \subset B(0,R)$). This way
\begin{equation}
\big| \mathcal{S}(\mathcal{P}_L) - \widehat{g \omega}(0) n_+^L \big| \leq C \widehat{g}(0) R^2 \sum_{p\in \mathcal{P}_L}|p|^2a_p^*a_p \leq C \widehat{g}(0) K_H^2 R^2 \ell^{-2} n_+.
\end{equation}
Finally, the remaining $\widehat{g\omega}(0) n_+^H$ is controlled by $\widehat{g}(0) n_+^H$.

2. To prove the second estimate, we introduce a cutoff version of the scattering function. Let $\chi$ be a smooth and radial function so that $0 \leq \chi \leq 1$ and $\chi(x) = 1$ for $|x|\leq 1/3$ and $\chi(x)=0$ for $|x|>1/2$. Let us define $\omega_c (x) = \omega(x) \chi(x/\ell)$, it has support inside $\Lambda_{\ell/2}$ and satisfies
\begin{align*}
-\Delta \omega_c = \frac{1}{2}g - \frac{a}{\ell^{3}} \left(\frac{\chi''}{|\cdot|}\right)(x\ell^{-1}).
\end{align*}
Denoting $U = \frac{\chi''}{|\cdot|}$, which is smooth and compactly supported, we have
\begin{align*}
p^2\widehat{\omega}_c(p) = \frac{1}{2} \widehat{g}(p) - a \widehat{U}(\ell p).
\end{align*}
Using that $\chi(x/\ell)=1$ for $x \in \supp g$ we have that $\widehat{g \omega} (0) = \widehat{g \omega_c}(0)$. The Plancherel formula then gives
\begin{align*}
\left|  \int g \omega_c - \frac{1}{8 \ell^3} \sum_{\substack{k \in \frac{\pi}{\ell}\mathbb Z^3\setminus\{0\}}} \frac{ \widehat{g}(k)^2}{2 |k|^2} \right|\leq \frac{1}{8 \ell^3} \sum_{\substack{k \in \frac{\pi}{\ell}\mathbb Z^3\setminus\{0\}}} a\left| \frac{\widehat{g}(k) \widehat{U}(k\ell)}{|k|^2} \right| \leq C a^2 \ell^{-1},
\end{align*}
where we used that $\|\widehat{g}\|_\infty \leq C a$.
This proves (\ref{eq.gomega0app0}). The bound (\ref{eq.gomega0app}) follows by using again that $\|\widehat{g}\|_\infty \leq C a$ and the definition of $\mathcal P_L^{\mathbb Z}$.
\end{proof}

Now we can prove Theorem \ref{thm.Q3}, combining $Q_3^{\rm{soft}}$ with the remaining quadratic part of the Hamiltonian after the Bogoliubov diagonalization. This process also uses a fraction of the diagonalized Hamiltonian.

\begin{proof}[Proof of Theorem \ref{thm.Q3}]
First of all, we can use Lemma \ref{lem.Q3loc2} to replace $Q_{3,L}^{\rm{sym}}(z)$ by $Q_3^{\rm{soft}}(z)$ defined in \eqref{Def. Q3soft}. In this expression we replace $a_k$ by $b_k$ using \eqref{def:bdiagonalization} and find
	\begin{equation}
		\begin{aligned}
		Q_3^{\mathrm{soft}}(z)&=\frac{1}{\vert \Lambda \vert} \sum_{ \substack{k \in \mathcal P_H \\p \in \mathcal P_L^{\mathbb{Z}}}} \frac{c(p,k) \widehat{g}(k) }{\sqrt{1-\alpha_k^2}}\big(   \overline{z}a_p^* a_{p-k} b_k +\hc \big)-\frac{1}{\vert \Lambda \vert} \sum_{ \substack{k \in \mathcal P_H \\p \in  \mathcal P_L^{\mathbb{Z}}}} \frac{c(p,k)\widehat{g}(k)\alpha_k}{\sqrt{1-\alpha_k^2}}\big(   \overline{z}a_p^* a_{p-k} b_{k}^*+\hc \big). \\ &=: \mathcal T_1 - \mathcal T_\alpha.
	\end{aligned}
	\end{equation}
The second term $\mathcal T_\alpha$ is an error, which we can bound as follows, using a Cauchy-Schwarz inequality,
\begin{align*}
\mathcal T_\alpha &= \frac{1}{\vert \Lambda \vert} \sum_{ \substack{k \in \mathcal P_H \\p \in  \mathcal P_L^{\mathbb{Z}}}} \frac{c(p,k)\widehat{g}(k)\alpha_k}{\sqrt{1-\alpha_k^2}}\big(   \overline{z}a_p^* a_{p-k} b_{k}^*+ \hc \big) \\
& \leq \frac{ C}{|\Lambda |} \sum_{\substack{k \in \mathcal P_H \\p \in  \mathcal P_L^{\mathbb{Z}}}} \frac{\widehat{g}(k)\alpha_k}{\sqrt{1-\alpha_k^2}} \big( \rho a k^{-2} K_H^{-2} K_\ell^{-4} |z|^2 b_k b_k^* + (\rho a)^{-1} k^2 K_H^2 K_\ell^4 a_p^* a_{p-k} a_{p-k}^* a_p  \big).
\end{align*}
We then use that $|\alpha_k| \leq C \rho_z a k^{-2} \leq C K_\ell^{1/4} \rho a k^{-2}$ and $|\mathcal P_L^{\mathbb Z}| \leq C K_H^3$ to obtain
\begin{align*}
\mathcal T_\alpha &\leq C\frac{N \rho^2 a^3 K_H K_\ell^{-4 +1/2} }{|\Lambda|}  \sum_{k \in \mathcal P_H} \frac{1}{k^4} (b_k^* b_k +1) + C \frac{a K_H^2 K_\ell^5}{|\Lambda|} \sum_{\substack{k \in \mathcal P_H \\p \in  \mathcal P_L^{\mathbb{Z}}}}  a_p^* a_p (a_{p-k}^* a_{p-k} +1)\\
&\leq C K_H^{-5} K_\ell^{5/2} \sum_{k \in \mathcal P_H} k^2 b_k^* b_k + C N \rho a \sqrt{\rho a^3} K_\ell^{-1} + C \mathcal M \frac{a K_H^2 K_\ell^5}{\ell} \frac{n_+^L n_+}{\mathcal M \ell^2},
\end{align*}
where in the last inequality we used $k^6 \geq K_H^6 \ell^{-6}$. Since $C K_H^{-5} K_\ell^{5/2} \leq K_H^{-1}$ and $k^2 \leq C D_k$, the first term is absorbed in a fraction of the Bogoliubov Hamiltonian, and the last term is absorbed in a fraction of the spectral gaps if $\mathcal M \leq C^{-1} \rho \ell^3 K_\ell^{-7} K_H^{-2}$, which is guaranteed by the assumptions on $\mathcal{M}$ and $K_H \geq K_{\ell}^4$.
We can now focus in $\mathcal T_1$. Using the remaining fraction of the diagonalized Hamiltonian, we complete a square to find
\begin{equation}\label{eq. completing the square}
\begin{aligned}
		\mathcal T_1 +(1-2 K_H^{-1}) \sum_{k\in \mathcal P_H} D_k b_k^* b_k&=\sum_{k\in \mathcal P_H}(1-2 K_H^{-1})D_k(b_k+ A_k)^*(b_k+A_k)-\sum_{k\in \mathcal P_H}(1-2K_H^{-1})D_k A_k^*A_k,
\end{aligned}
\end{equation}
where
	\[A_k=\frac{z\widehat{g}(k)}{\vert \Lambda\vert (1-2K_H^{-1})D_k\sqrt{1-\alpha_k^2}}\sum_{p\in \mathcal{P}^{\mathbb{Z}}_L}c(p,k) a_{p-k}^*a_p.\]
The first term in \eqref{eq. completing the square} is positive and can be dropped for a lower bound. We are left with a term in $A_k^*A_k$, which we can rewrite in normal order as
\begin{equation}\label{eq. the Ak's}
\begin{aligned}
	(1- 2 K_H^{-1}) \sum_{k\in \mathcal{P}_H}D_k A_k^*A_k=	\sum_{k\in \mathcal{P}_H}\frac{\rho_z\widehat{g}(k)^2}{ \vert\Lambda\vert  (1-2K_H^{-1})D_k(1-\alpha_k^2)}\Big(&\sum_{p,s\in \mathcal{P}_L^{\mathbb{Z}}} c(p,k) c(s,k) a_p^*[a_{p-k},a_{s-k}^*]a_s \\ &+\sum_{p,s\in \mathcal{P}_L^{\mathbb{Z}}}c(p,k) c(s,k) a_p^* a_{s-k}^*a_{p-k}a_s\Big).
\end{aligned}
\end{equation}
We call the two terms of above $\mathcal{T}_{c}$ for the commutator term and $\mathcal{T}_{0}$ for the other one, so that
\begin{equation}
\mathcal T_1 +(1-2 K_H^{-1}) \sum_{k\in \mathcal P_H} D_k b_k^* b_k \geq - \mathcal T_{c} - \mathcal{T}_0.
\end{equation}
We start by estimating the main term $\mathcal{T}_{c}$, and then we bound the error term $\mathcal{T}_0$.

\textbf{Commutator term $\mathcal{T}_{c}$.}
Recall that the commutation $[a_p,a_q^{*}]=\delta_{p,q}$ only applies when $p,q\in \frac{\pi}{\ell}\mathbb{N}_0^3$. In the above commutators, due to the sets on which we sum, this may not be the case. We can however use that $a_p =a_{p^+ }$ where $p^+=(|p_1|,|p_2|,|p_3|)$, and deduce that
\begin{equation}\label{eq.commu}
 [a_{p-k},a_{s-k}^{*}] \neq 0 \Leftrightarrow \Big( \, p_j = s_j \quad \text{or} \quad 2k_j = p_j + s_j, \quad \forall j = 1,2,3.  \Big) 
\end{equation}
For $p$, $s \in \mathcal P_L$, the second case in \eqref{eq.commu} implies $|k_j| \leq K_H \ell^{-1}$. Therefore, if $p \neq s$,
\[ [a_{p-k},a_{s-k}^{*}] \neq 0 \Rightarrow \Big( \, |k_j| \leq K_H \ell^{-1} \quad \text{for some } j \, \Big)\]
and we deduce
\begin{align}\nonumber
\mathcal{T}_c &=\sum_{k\in \mathcal{P}_H}\frac{\rho_z\widehat{g}(k)^2}{\vert\Lambda\vert  (1-2 K_H^{-1})D_k(1-\alpha_k^2)}\sum_{p,s\in \mathcal{P}_L^{\mathbb{Z}}} c(p,k)c(s,k) a_p^*[a_{p-k},a_{s-k}^*]a_s \\
	 &\leq \sum_{k\in \mathcal{P}_H}\frac{\rho_z\widehat{g}(k)^2}{ \vert\Lambda\vert (1-2 K_H^{-1})D_k(1-\alpha_k^2)}\sum_{p\in \mathcal{P}_L^{\mathbb{Z}}} c(p,k)^2 a_p^* a_p  + \sum_{k\in \mathcal{P}_H}\frac{C \rho_z\widehat{g}(k)^2 \one_{ \lbrace \vert k_1\vert\leq K_H\ell^{-1}\rbrace} }{\vert\Lambda\vert (1-2 K_H^{-1})D_k(1-\alpha_k^2)}\sum_{p\in \mathcal{P}_L^{\mathbb{Z}}} a_p^* a_{s_{p,k}} ,
\end{align}
where the components of $s_{p,k}$ are either equal to $p_j$ or $2k_j - p_j$. In any case, we can always bound the last $p$-sum by $n_+$ using a Cauchy-Schwarz inequality. We also use $C D_k(1-\alpha_k^2) \geq k^2$ to get
\begin{equation} \label{eq.Tc1}
\mathcal T_c \leq \big( 1 + CK_H^{-1} \big) \frac{1}{|\Lambda |} \sum_{k\in \mathcal{P}_H}\frac{\rho_z\widehat{g}(k)^2}{k^2} \sum_{p\in \mathcal{P}_L^{\mathbb{Z}}}c(p,k)^2 a_p^* a_p + \frac{C}{|\Lambda|} \sum_{k\in \mathcal{P}_H} \frac{\rho_z \widehat{g}(k)^2}{k^2} \one_{ \lbrace \vert k_1\vert\leq K_H\ell^{-1}\rbrace} n_+.
\end{equation}
We recall that the normalization coefficients $c(p,k) \in [\frac{1}{\sqrt{8}},\sqrt{8}]$ are defined in \eqref{def. of cqk}. We write them as $c(p,k) = \frac{c_{p-k}}{c_k c_p}$ with the notation $c_k = c_{k_1} c_{k_2} c_{k_3}$. Note that the $c_p$'s (defined in \eqref{def:neubasis}) are such that 
\begin{equation}\label{eq.sumcomplete}
\sum_{p \in \mathcal P_L^{\mathbb Z}} c_p^{-2} a_p^* a_p = \sum_{p \in \mathcal P_L} a_p^* a_p = n_+^L.
\end{equation}
We can also bound $c_{p-k}^2 \leq 8$ and we deduce
\begin{equation}\label{eq.Tcwithcoeff}
\mathcal T_c \leq \big( 1 + CK_H^{-1} \big) \frac{1}{|\Lambda |} \sum_{k\in \mathcal{P}_H}\frac{8 \rho_z\widehat{g}(k)^2}{c_k^2 k^2} n_+^L + \frac{C}{|\Lambda|} \sum_{k\in \mathcal{P}_H} \frac{\rho_z \widehat{g}(k)^2}{k^2} \one_{ \lbrace \vert k_1\vert\leq K_H\ell^{-1}\rbrace} n_+.
\end{equation}
Similarly as in \eqref{eq.sumcomplete}, we can complete the $k$-sum to $\mathcal P_H^{\mathbb Z}$ up to an extra factor $c_{k}^{-2}$. Moreover, $c_k^4=8^2$ unless at least one of the components $k_j$ vanishes. Therefore, the terms for which $c_k^4 \neq 8^2$ can be controlled by the last term in \eqref{eq.Tcwithcoeff}, and we obtain
\begin{equation}
\mathcal T_c \leq \big( 1 + CK_H^{-1} \big) \frac{1}{|\Lambda |} \sum_{k\in \mathcal{P}_H^{\mathbb{Z}}}\frac{\rho_z\widehat{g}(k)^2}{8 k^2} n_+^L + \frac{C}{|\Lambda|} \sum_{k\in \mathcal{P}_H} \frac{\rho_z \widehat{g}(k)^2}{k^2} \one_{ \lbrace \vert k_1\vert\leq K_H\ell^{-1}\rbrace} n_+.
\end{equation}
In the first term we use point 2. of Lemma \ref{lem:gomega.approx} to replace the $k$-sum by $2 \widehat{g\omega}(0) \leq C \rho a \leq C \ell^{-2} K_\ell^{2}$. The second term is bounded in Lemma~\ref{lem:Q3tech} below, and we get
\begin{equation}\label{ine:q2}
\mathcal{T}_c \leq 2 \rho_z \widehat{g\omega}(0)n_+ + C\Big( K_\ell^2 K_H^{-1} + (\rho a^3)^{1/2}K_H^2 \Big)K_\ell^{1/4} \ell^{-2} n_+. 
\end{equation}
By point 1. of Lemma \ref{lem:gomega.approx}, the first term of above is precisely the remaining quadratic term we want to cancel in Theorem \ref{thm.Q3}. The second one is absorbed in spectral gaps when $K_H \geq K_\ell^4$ and $K_H^2 K_\ell^{1/4} \leq C^{-1} (\rho a^3)^{- \frac{1}{2}}$. 

\textbf{The error term $\mathcal{T}_0$.}
It only remains to control $\mathcal{T}_0$. We use similar bounds as for $\mathcal{T}_c$ and a Cauchy Schwarz inequality and find
\begin{equation}
\begin{aligned}
\mathcal{T}_0 &= \sum_{k \in \mathcal P_H} \frac{\rho_z \widehat{g}(k)^2}{ |\Lambda| (1-2K_H^{-1}) D_k (1-\alpha_k^2)}\sum_{p,s\in\mathcal{P}_L^{\mathbb{Z}}} c(p,k)c(s,k) a_p^*a_{s-k}^*a_{p-k}a_s\\
&\leq \frac{C}{\vert \Lambda\vert} \sum_{k\in\mathcal{P}_H}\frac{\rho_z \widehat{g}(0)^2}{ k^2}\sum_{p,s\in\mathcal{P}_L^{\mathbb{Z}}}a_p^*a_{s-k}^*a_{s-k}a_p\leq C \frac{  \widehat{g}(0) K_{ \ell}^3}{|\Lambda| K_H^2} \sum_{k\in\mathcal{P}_H} \sum_{p,s\in\mathcal{P}_L^{\mathbb{Z}}}a_p^*a_{s-k}^*a_{s-k}a_p,
\end{aligned}
\end{equation}
where in the second inequality we used $|k| \geq K_H \ell^{-1}$ and $\rho_z \leq 2 \rho$. The $k$-sum can be bounded by $n_+$, the $p$-sum by $n_+^L$, and remains the cardinal of $\mathcal P_L^{\mathbb Z}$, i.e.
\begin{equation}
\mathcal{T}_0 \leq \frac{C}{|\Lambda|} \widehat g(0) K_{\ell}^3 K_H n_+ n_+^L.
\end{equation}
This is absorbed in spectral gaps under the condition $\mathcal M \leq C^{-1}\rho \ell^3 K_\ell^{-5} K_H^{-1}$, for $C$ large enough.
\end{proof}

\begin{lemma}\label{lem:Q3tech}
Under the assumptions of Theorem \ref{thm.Q3} we have
\[\frac{1}{|\Lambda|} \sum_{k\in \mathcal{P}_H} \frac{\rho \widehat{g}(k)^2 }{k^2}\one_{\{\vert k_1\vert\leq K_H\ell^{-1}\}}  \leq C \ell^{-2} (\rho a^3)^{1/2}K_H^2.\]
\end{lemma}

\begin{proof}
We first remove the very high momenta from the sum, for $|k| > K_0 \ell^{-1}$, with $K_0>0$ to be chosen later,
\begin{equation}
\frac{1}{|\Lambda|} \sum_{|k| > \frac{K_0}{\ell} } \frac{\rho \widehat{g}(k)^2 }{k^2}\one_{\{\vert k_1\vert\leq K_H\ell^{-1}\}} 
 \leq \frac{C \rho \ell^2 \| g \|_2^2 }{K_0^2} \leq C \frac{\rho \ell^{4}}{K_0^{2}a^{3}} = C \ell^{-2} K_\ell^6 (\rho a^3)^{-2} K_0^{-2}
\end{equation}
where in the last inequality we used that $\|g\|_2^2 \leq \|g\|_\infty \|g\|_1 \leq C \ell^2 a^{-3}$. We now deal with the rest
\begin{align*}
\frac{1}{|\Lambda|} \sum_{|k| \leq \frac{K_0}{\ell} } \frac{\rho \widehat{g}(k)^2 }{k^2}\one_{\{\vert k_1\vert\leq K_H\ell^{-1}\}} \leq C \rho a^2 \ell^{-1} K_H \log K_0 = C \ell^{-2} (\rho a^3)^{1/2}  K_H K_\ell^{-1} \log K_0.
\end{align*}
Choosing $K_0 = K_\ell^3 (\rho a^3)^{-5/4}$ and using that $\log K_0 \leq K_H$ concludes the proof.
\end{proof}

%%%%%%%%%%%%%%%%%%%%%%%%%%%%%%%%%%%%%%%%%%%%%%%%%%
%%%%%%%%%%%%%%%%%%%%%%BOGOLIUBOVINT%%%%%%%%%%%%%%%%%%%%%%%%%%

\appendix
\section{Estimates on the LHY terms}

\begin{lemma}\label{lem.integral}
	There exists a $C>0$ such that, if $|z|^2 \leq K_\ell^{1/4} N \leq C \rho\ell^3 K_\ell^{1/4}$, we have 
	\begin{align*}
		\rho_z^2 \widehat{g \omega}(0) + \frac{1}{\vert \Lambda \vert} \sum_{p \in \Lambda_+^*} \Big( &\sqrt{\tau(p)^2 + 2 \rho_z \widehat g(p) \tau(p)} - \tau(p)- \rho_z \widehat g(p) \Big) = 8\pi (\rho_z a)^{5/2} \frac{128}{15 \sqrt \pi} + \mathcal O((\rho a)^{5/2} K_\ell^{-1/4}), 
	\end{align*}
We recall that $\tau(p)$ is defined in (\ref{eq:tau}).
\end{lemma}

\begin{proof}
Using (\ref{eq.gomega0app0}) and that $|p^{-2} - \tau_p^{-1}| \leq C p^{-4} \ell^{-2} (1 + K_H \one_{p\in \mathcal P_H})$, we can rewrite the left-hand side of the above equation as
\begin{align*}
		\sum_{p \in \Lambda_+^*} \Big( &\sqrt{\tau(p)^2 + 2 \rho_z \widehat g(p) \tau(p)} - \tau(p)- \rho_z \widehat g(p) + \frac{(\rho_z \widehat{g}(p))^2}{2\tau(p)}\Big) + \mathcal O(|\Lambda| (\rho a)^{5/2} K_\ell^{-1/2}),
\end{align*}
where the error can be absorbed in $\mathcal E$. To estimate the sum, we define $G(t) = \sqrt{1+2t} - 1 - t + t^2/2 \geq 0$, which is such that  $x G(y/x) = \sqrt{x^2 + 2 xy} - x -y - x/(2y)$. 
Let us introduce a cut-off $1<  K < K_H$ which we will choose at the end. Using that $G(t)\leq C t^3$, we have for $\rho_z \leq K_\ell^{1/4}\rho$,
\begin{align*}
\sum_{|p| > K \ell^{-1}} \tau(p) G\Big(\frac{\rho_z \widehat{g}(p)}{\tau(p)}\Big) + \sum_{|p| > K \ell^{-1}} p^2 G\Big(\frac{8 \pi\rho_z a}{p^2}\Big) \leq C K_\ell^{3/4} (\rho a)^3 \sum_{|p| > K \ell^{-1}} \frac{1}{p^4} \leq C \ell^3 (\rho a)^{5/2} K_\ell^{7/4} K^{-1},
\end{align*}
where we recall that the sums are over $p \in 2\pi \ell^{-1} \mathbb{Z}^3$.
Let us now deal with $|p| \leq K \ell^{-1}$. Note that $G(t) \leq C t^2$ and $|G'(t)| \leq C(1+ t)$ for some $C>0$ and all $t\geq0$, so that
\begin{align*}
\left| (\tau(p)-p^2) G\Big(\frac{\rho_z \widehat{g}(p)}{\tau(p)}\Big)\right|
	&\leq C K_\ell^{1/2} (\rho a)^2 \ell^{-2} p^{-4} \\
p^2 \left| G\Big(\frac{\rho_z \widehat{g}(p)}{\tau(p)}\Big)- G\Big(\frac{8\pi a \rho_z}{p^2}\Big) \right| & \leq  C K_\ell^{1/4} (\rho a) p^2 (1+ K_\ell^{1/4} \rho a p^{-2}) (p^{-4}\ell^{-2} + R^2) 
\end{align*}
where we used that $|\tau(p)-p^2| \leq C \ell^{-2}$ for $|p| \leq K \ell^{-1}$ and that $|\widehat{g}(p) - \widehat{g}(0)| \leq R^2 \widehat{g}(0)|p|^2$. We obtain that
\begin{align*}
\left|\sum_{|p| \leq K \ell^{-1}} \tau(p) G\Big(\frac{\rho_z \widehat{g}(p)}{\tau(p)}\Big) - p^2 G\Big(\frac{8 \pi\rho_z a}{p^2}\Big) \right| 
	&\leq C K_\ell^{1/4} (\rho a)  \left( \ell^{-2} R^2 K^5+  \rho a R^2 K_\ell^{1/4} K^3 + K + K_\ell^{9/4} \right) \\
	&\leq C \ell^3 (\rho a)^{5/2} K_\ell^{-\frac{11}{4}} \left( \rho a^3 K^5 K_\ell^{-2} + \rho a^3 K^3 K_\ell^{1/2} + K + K_\ell^{9/4} \right)
\end{align*}
where we used that $R \leq C a$. Choosing $K = K_\ell^2$ and using that $K_\ell \leq C (\rho a^3)^{1/10}$ we can absorb the error terms in $\mathcal E$. Next, we approximate the sum by the integral, recalling that $\Lambda^* = \pi \ell^{-1} \mathbb{N}_0^{3}$, one easily checks that
\begin{align*}
\left| \frac{\pi^3}{\ell^3} \sum_{p \in \Lambda_+^*}  p^2 G\Big(\frac{8 \pi\rho_z a}{p^2}\Big)  - 8 \int_{\mathbb{R}_+^{3}} p^2 G\Big(\frac{8 \pi\rho_z a}{p^2}\Big) dp \right| \leq C (\rho_z a)^3 \ell^{-1} \leq C K_\ell^{-1/4} (\rho a)^{5/2}.
\end{align*}
Finally, it is a standard result that $\int_{\mathbb{R}^{3}} p^2 G\Big(\frac{8 \pi\rho_z a}{p^2}\Big) = - 64 \pi^4 \frac{128}{15\sqrt{\pi}} (\rho_z a)^{5/2}$ (see for instance \cite{FS}).
\end{proof}

Recall that
\begin{equation}
\tilde{D}_p(z) = \begin{cases}
D_p(z) &\text{if} \quad p \notin \mathcal P_H,\\
K_H^{-1} D_p(z) &\text{if} \quad p \in \mathcal P_H.
\end{cases},\qquad D_p(z):=\sqrt{\tau(p)^2+2\tau(p)|z|^2 \ell^{-3}\widehat{g}(p)}.
\end{equation}

\begin{lemma}\label{lem.thermal.approx}
If $\nu < 2 \eta < 1/5$ and $|z|^2 \leq C K_\ell^{1/4} \rho \ell^3 $, then we have
\[ T\sum_{p \in \Lambda^*_+}\log(1-e^{-\frac{1}{T}\tilde{D}_p(z)}) \geq T\sum_{p \in \Lambda^*_+}\log(1-e^{-\frac{1}{T}\omega_p(z)})-  C |\Lambda| (\rho a)^{3}, \]
where $\omega_p(z) = \sqrt{p^4 + 16 \pi a |z|^2\ell^{-3} p^2}$.
\end{lemma}

\begin{proof}
First note that $D_p(z) \geq p^2/4$ which gives
\begin{align*}
T \sum_{p \in \mathcal P_H}\log(1-e^{-\frac{1}{T}\tilde{D}_p(z)})  \geq - T \sum_{p \in \mathcal P_H} e^{-p^2/(4T K_H)} \geq - C T^{5/2} \ell^3 K_H^{3/2} e^{-C K_H/(T\ell^2)},
\end{align*}
which is smaller than the expected error since $K_{H}/(T \ell^2)$ is big. Similarly, we can neglect the term $T \sum_{p \in \mathcal P_H}\log(1-e^{-\frac{1}{T}\omega_p(z)})$.
Now note that there is a constant $C>0$ such that, for all $p \in \mathcal P_L$,
\begin{equation}\label{eq.Dp-omegap}
\big| D_p - \omega_p \big| \leq C \Big( 1+ \frac{\sqrt{|z|^2\ell^{-3} a}}{|p|} \Big) |p^2 - \tau(p)| \leq C \frac{\sqrt{\rho a} K_\ell^{\frac{1}{8}}}{|p|} \frac{1}{\ell^2}
\end{equation}
Moreover, denoting $G(x) = \log( 1-e^{-x})$, we have $G'(x) = \frac{1}{e^{x}-1}$ and a Taylor expansion gives
\begin{align*}
T\sum_{p \in \mathcal P_L}\log(1-e^{-\frac{1}{T}{D}_p(z)}) - T\sum_{p\in \mathcal P_L}\log(1-e^{-\frac{1}{T}\omega_p(z)}) 
	&\geq -C \frac{\sqrt{\rho a} K_\ell^{1/8}}{\ell^2}\sum_{p \in \mathcal P_L} \frac{ 1 }{|p|(e^{\frac{p^2}{4T}}-1)} \\
	&\geq - C\frac{\sqrt{\rho a} K_\ell^{1/8}}{\ell^2} T^{2} \ell^{3} \log(T^{1/2}\ell) \\
	&\geq - C |\Lambda| \log (\rho a^3) (\rho a^{3})^{7/2+15\eta/8 - 2\nu}
\end{align*}
and the result follows from the assumptions on $\eta$ and $\nu$.
 \end{proof}

\section[Reduction to small boxes]{Reduction to small boxes: proof of Theorem \ref{TL.thm}}\label{sec:redtoboxes}

In this section we prove Theorem \ref{TL.thm} in the thermodynamic limit using the lower bound on small boxes from \Cref{main.thm}. We follow the proof of \cite[Theorem 1.1]{HHNST}. It applies to our case, but we want to keep track of our parameters. The limit (\ref{eq:def_f}) is independent of the choice of $N$ and $L$, we can therefore choose $L$ such that $L/\ell$ is an integer and define $M = (L/\ell)^3$. We can now divide $\Lambda_L$ into $M$ translates of $\Lambda_\ell$. From \cite[Eq. (9.6)]{HHNST}, we obtain immediately that, for all $\mu \in \mathbb{R}$,
\begin{align} \label{eq:GC-to-C}
F(L,N) \geq - T M \log \sum_{n=0}^N e^{-\frac{1}{T} (F(\ell,n) - \mu n)} + \mu \rho L^3.
\end{align}

For $n \leq 20 \rho \ell^3$, we use \Cref{main.thm} to obtain
\begin{equation*}
 F(\ell,n) - \mu n \geq F_{\rm{Bog}}(\ell, n) - \mu n - C \ell^3 (\rho a)^{5/2} (\rho a^3)^{\eta/4},
 \end{equation*}
where we recall that
\begin{align*} 
F_{\mathrm{Bog}}(\ell,n) = 4\pi \rho_{n,\ell}^2 a\ell^3\Big( 1+ \frac{128}{15 \sqrt \pi} \sqrt{ \rho_{n,\ell} a^3} \Big) + T \sum_{p \in \Lambda_+^*} \log \Big( 1- e^{- \frac{1}{T}\sqrt{p^4 + 16\pi a \rho_{n,\ell} p^2}}\Big).
\end{align*}
It is easy to check that it is convex in $n$ for $\rho a^3$ small enough using to the following lemma.
\begin{lemma}\label{lem:convexity2}
There exists $C >0$ such that for all $\rho \geq 0$, $\ell,T>0$,
\begin{align}
0 \leq \frac{\partial}{\partial \rho} T \sum_{p \in \Lambda_+^*} \log(1-e^{-T^{-1} \sqrt{p^4 + 16\pi a \rho p^2}}) \leq C T^{5/2} \ell^3, \label{eq:lem:convexity2_1}\\
0 \leq - \frac{\partial^2}{\partial \rho^2} T \sum_{p \in \Lambda_+^*} \log(1-e^{-T^{-1} \sqrt{p^4 + 16\pi a \rho p^2}}) \leq C T^{5/2} \ell^3. \label{eq:lem:convexity2_2}
\end{align}
\end{lemma}
Choosing $\mu$ so that $\frac{\partial}{\partial n}F_{\mathrm{Bog}}(\ell,\rho \ell^3) = \mu$, it follows that 
\begin{align} \label{eq:est_F_prop:convex}
F_{\mathrm{Bog}}(\ell,n) - \mu n \geq F_{\mathrm{Bog}}(\ell,\rho\ell^3) - \mu \rho \ell^3
\end{align}
for all $n\geq 1$ and one can easily check from (\ref{eq:lem:convexity2_1}) that
\begin{align} \label{eq:est_mu}
\left| \mu(\rho,a,T,\ell) - 8\pi a\rho \right| \leq C \rho a (\rho a^3)^{1/2-5\nu/2}
\end{align}
which gives
\begin{equation}\label{eq:n<20rhol3}
 F(\ell,n) - \mu n \geq F_{\rm{Bog}}(\ell, \rho \ell^3) - \mu \rho \ell^3 - C \ell^3 (\rho a)^{5/2} (\rho a^3)^{\eta/4}. 
 \end{equation}
For $n> 20 \rho \ell^3$, denoting $n_0 = \lfloor 5\rho\ell^3 \rfloor $ and $r_0 = n - \left\lfloor \frac{n}{n_0} \right\rfloor n_0$, we use the subadditivity of the free energy $F(\ell,n)$ to obtain
\begin{align*}
F(\ell,n) -\mu n 
	&\geq \left \lfloor \frac{n}{n_0} \right\rfloor (F(\ell, n_0) -  \mu n_0) + (F(\ell, r_0) - \mu r_0) \\
	&\geq \left \lfloor \frac{n}{n_0} \right\rfloor (F_{\mathrm{Bog}}(\ell, n_0) -  \mu n_0) + (F_{\mathrm{Bog}}(\ell, r_0) - \mu r_0) - C  \left \lfloor \frac{n}{n_0} \right\rfloor \ell^3 (\rho a)^{5/2} (\rho a^3)^{\eta/4},
\end{align*}
where we used \Cref{main.thm} to obtain the second inequality, since $n_0, r \leq 20 \rho \ell^3$.
To estimate the first term, we use the bound (\ref{eq:est_mu}) on $\mu$ and that, for $T\leq C \rho a (\rho a^3)^{-\nu}$,
\begin{align*}
F_{\mathrm{Bog}}(\ell, n_0) - \mu n_0 \geq 4 \pi a  \big( \frac{n_0}{\ell^3} \big)^2 \ell^3 - 8\pi a \rho n_0 - C  a \rho^2\ell^3 (\rho a^3)^{1/2-5\nu/2} \geq 10 \pi a \rho^2 \ell^3,
\end{align*} 
for $\rho a^3$ small enough.
Therefore, estimating $F_{\mathrm{Bog}}(\ell, r_0) $ with (\ref{eq:est_F_prop:convex}) we obtain for $n>n_0$,
\begin{align}
F(\ell,n) - \mu n 
	&\geq F_{\mathrm{Bog}}(\ell, \rho \ell^3) - \mu \rho\ell^3 + 10 \left  \lfloor \frac{n}{n_0} \right\rfloor \pi a \rho^2 \ell^3 \nonumber\\
	&\geq F_{\mathrm{Bog}}(\ell, \rho \ell^3) - \mu \rho\ell^3 + n a \rho. \label{eq:n>20rhol3}
\end{align}
Hence, from (\ref{eq:GC-to-C}) and the bounds \eqref{eq:n<20rhol3} and \eqref{eq:n>20rhol3}, we obtain
\begin{align*}
\frac{1}{L^3}F(L,N) - \frac{1}{\ell^3}F_{\mathrm{Bog}}(\ell,\rho\ell^3) 
	&\geq - \frac{T}{\ell^3} \log \left(\sum_{n\leq 20 \rho \ell^3} e^{C \frac{1}{T}\ell^3 (\rho a)^{5/2} (\rho a^3)^{\eta/4}} + \sum_{n> 20 \rho \ell^3} e^{-n \frac{\rho a}{T}} \right) \\
	&\geq - (\rho a)^{5/2} (\rho a^3)^{-\nu} K_\ell^{-3} \log (n_0 e^{C \ell^3 \frac{(\rho a)^{5/2}}{T} (\rho a^3)^{\eta/4}} + e^{-n_0\frac{\rho a}{T}}/(1-e^{-\frac{\rho a}{T}})) \\
	&\geq - C (\rho a)^{5/2} (\rho a^3)^{\eta/4}+C\frac{T}{\ell^3}\log(\rho a^3).
\end{align*}
where we used that $e^{-n_0\frac{\rho a}{T}}/(1-e^{-\frac{\rho a}{T}})\leq 1$ and that $n_0 \simeq (\rho a^3)^{-1/2} K_\ell^3$. Now from \cite[Lemma 9.1]{HHNST}, we obtain
\begin{align*}
&\left| \frac{T}{\ell^3} \sum_{p \in \frac{\pi}{\ell} \mathbf N_0^3} \log \Big( 1- e^{- \frac{1}{T}\sqrt{p^4 + 16\pi a \rho p^2}}\Big)  - \frac{T^{5/2}}{(2\pi)^3} \int_{\mathbb{R}^{3}} \log\left(1-e^{-\sqrt{p^4 + \tfrac{16 \pi \rho a}{T} p^2}}\right)\dd p \right| \\ 
&\leq C (T\ell^2)^{-1/2} T^{5/2} \leq C (\rho a)^{5/2}(\rho a^3)^{- 2 \nu} K_\ell^{-1}.
\end{align*}
From which Theorem \ref{TL.thm}  follows  since $\nu < \eta/3$.
\qed

\begin{proof}[Proof of Lemma \ref{lem:convexity2}]
Let us define $w_p(\rho) = \sqrt{p^4 + 16\pi \rho a T^{-1} p^2}$ and $G(x) = \log (1-e^{-x})$. For all $p$ we have
\begin{align*}
w_p'(\rho) = \frac{8\pi a T^{-1} p^2}{w_p(\rho)}, \qquad w_p''(\rho) = - \frac{\left(8\pi a T^{-1} p^2\right)^2}{w_p(\rho)^3}
\end{align*}
and $G'(x) = 1/(e^{x}-1)$, $G''(x) = - e^x / (e^x-1)^2$. We deduce that
\begin{align*}
 0 \leq \frac{\partial}{\partial \rho} G(w_p(\rho)) = \frac{8\pi a T^{-1} p^2}{w_p(\rho)^2} \frac{w_p(\rho)}{e^{w_p(\rho)}-1} \leq C e^{-p^2}
\end{align*}
and that
\begin{align*}
0\geq \frac{\partial^2}{\partial \rho^2} G(w_p(\rho)) &= - \frac{\left(8\pi a T^{-1} p^2\right)^2}{w_p(\rho)^4} \frac{w_p^2(\rho)}{(e^{w_p(\rho)}-1)^2} \left(e^{w_p(\rho)} + \frac{e^{w_p(\rho)}-1}
	{w_p(\rho)}\right) \\
		&\geq - C e^{-w_p(\rho)}(w_p(\rho) + 1)^2 \geq - C e^{-p^2/2}(p^2 + 1)^2,
\end{align*}
where we used that $w_p(\rho)^2 \geq 16\pi a T^{-1} p^2$, that $(e^{x}-1)/x \leq e^{x}$ and that $x/(e^{x}-1)) \leq C e^{-x}(x+1)$. The result follows by estimating the associated Riemann sum.
\end{proof}

\paragraph{Acknowledgements}
This work was partially funded by the Deutsche Forschungsgemeinschaft (DFG, German Research Foundation) – Project-ID 470903074 – TRR 352.
S.F., L.J., L.M., and M.O. were partially supported by the grant 0135-00166B from the Independent Research Fund Denmark and by the Villum Centre of Excellence for the Mathematics of Quantum Theory (QMATH) with Grant No.10059. M.O. was partially supported by the European Union’s HORIZON-MSCA-2022-PF-01 grant agreement, project number: 101103304 (UniBoGas). Furthermore, S.F., L.M. and M.O. were partially funded by the European Union. The author T.G. was partially funded by the grant HE ERC MATCH-G.A. n. 101117299-CUP: D13C23002820006, Macroscopic properties of interacting bosons: a unified approach to the Thermodynamic Challenge, as well as the mathematics area of GSSI. Views and opinions expressed are however those of the authors only and do not necessarily reflect those of the European Union or the European Research Council. Neither the European Union nor the granting authority can be held responsible for them.

\bibliographystyle{plain}

%\bibliography{positiveref}

\begin{thebibliography}{10}

\bibitem{agerskov2024ground}
Johannes Agerskov, Robin Reuvers, and Jan~Philip Solovej.
\newblock Ground state energy of dilute bose gases in 1d, 2024.

\bibitem{BCC}
Giulia Basti, Cristina Caraci, and Serena Cenatiempo.
\newblock Energy expansions for dilute {B}ose gases from local condensation
  results: a review of known results.
\newblock In {\em Quantum mathematics {II}}, volume~58 of {\em Springer INdAM
  Ser.}, pages 199--227. Springer, Singapore, [2023] \copyright2023.

\bibitem{newupperbound}
Giulia Basti, Serena Cenatiempo, Alessandro Giuliani, Giulio~Pasqualetti
  Alessandro~Olgiati, and Benjamin Schlein.
\newblock Upper bound for the ground state energy of a dilute bose gas of hard
  spheres, 2023.

\bibitem{BCS}
Giulia Basti, Serena Cenatiempo, and Benjamin Schlein.
\newblock A new second-order upper bound for the ground state energy of dilute
  {B}ose gases.
\newblock {\em Forum of Mathematics, Sigma}, 9:e74, 2021.

\bibitem{bocdeuchert}
Chiara Boccato, Andreas Deuchert, and David Stocker.
\newblock Upper bound for the grand canonical free energy of the bose gas in
  the gross-pitaevskii limit, 2023.

\bibitem{zbMATH07681333}
Chiara Boccato and Robert Seiringer.
\newblock The {Bose} gas in a box with {Neumann} boundary conditions.
\newblock {\em Ann. Henri Poincar{\'e}}, 24(5):1505--1560, 2023.

\bibitem{BSol}
Birger Brietzke and Jan Solovej.
\newblock The second-order correction to the ground state energy of the dilute
  bose gas.
\newblock {\em Ann. Inst. Henri Poincar{\'e}}, 21, 01 2020.

\bibitem{capdeuchert}
Marco Caporaletti and Andreas Deuchert.
\newblock Upper bound for the grand canonical free energy of the bose gas in
  the gross-pitaevskii limit for general interaction potentials, 2023.

\bibitem{zbMATH07201541}
Andreas Deuchert, Simon Mayer, and Robert Seiringer.
\newblock The free energy of the two-dimensional dilute {Bose} gas. {I}:
  {Lower} bound.
\newblock {\em Forum Math. Sigma}, 8:74, 2020.
\newblock Id/No e20.

\bibitem{dyson}
Freeman~J. Dyson.
\newblock Ground-state energy of a hard-sphere gas.
\newblock {\em Phys. Rev.}, 106(1):20--26, apr 1957.

\bibitem{greenbook}
Robert~Seiringer Elliott H.~Lieb, Jan Philip~Solovej and Jakob Yngvason.
\newblock {\em The {M}athematics of the {B}ose {G}as and its {C}ondensation}.
\newblock Oberwolfach Seminars. Birkh{\"a}user Basel, 2005.

\bibitem{ESY}
L{\'{a}}szl{\'{o}} Erd{\H{o}}s, Benjamin Schlein, and Horng-Tzer Yau.
\newblock Ground-state energy of a low-density bose gas: A second-order upper
  bound.
\newblock {\em Phys. Rev. A}, 78(5), nov 2008.

\bibitem{2DLHY}
S{\o}ren Fournais, Theotime Girardot, Lukas Junge, Leo Morin, and Marco
  Olivieri.
\newblock The ground state energy of a two-dimensional {Bose} gas.
\newblock {\em Commun. Math. Phys.}, 405(3):104, 2024.
\newblock Id/No 59.

\bibitem{FS}
S{\o}ren Fournais and Jan~Philip Solovej.
\newblock {The energy of dilute Bose gases}.
\newblock {\em Ann. of Math.}, 192(3):893 -- 976, 2020.

\bibitem{FS2}
S{\o}ren Fournais and Jan~Philip Solovej.
\newblock The energy of dilute {Bose} gases. {II}: {The} general case.
\newblock {\em Invent. Math.}, 232(2):863--994, 2023.

\bibitem{zbMATH05594072}
Alessandro Giuliani and Robert Seiringer.
\newblock The ground state energy of the weakly interacting {Bose} gas at high
  density.
\newblock {\em J. Stat. Phys.}, 135(5-6):915--934, 2009.

\bibitem{HHNST}
Florian Haberberger, Christian Hainzl, Phan~Th{\`a}nh Nam, Robert Seiringer,
  and Arnaud Triay.
\newblock The free energy of dilute bose gases at low temperatures, 2023.

\bibitem{haberberger2024upper}
Florian Haberberger, Christian Hainzl, Benjamin Schlein, and Arnaud Triay.
\newblock Upper bound for the free energy of dilute bose gases at low
  temperature, 2024.

\bibitem{HST}
Christian Hainzl, Benjamin Schlein, and Arnaud Triay.
\newblock Bogoliubov theory in the gross-pitaevskii limit: a simplified
  approach, 2022.

\bibitem{LHY}
Tsung-Dao Lee, Kerson Huang, and Chen-Ning Yang.
\newblock Eigenvalues and eigenfunctions of a bose system of hard spheres and
  its low-temperature properties.
\newblock {\em Phys. Rev.}, 106(6):1135--1145, jun 1957.

\bibitem{LewNamSerSol-15}
Mathieu Lewin, Phan~Thanh Nam, Sylvia Serfaty, and Jan~Philip Solovej.
\newblock Bogoliubov spectrum of interacting {B}ose gases.
\newblock {\em Comm. Pure Appl. Math.}, 68(3):413--471, 2015.

\bibitem{zbMATH03224473}
Elliott~H. Lieb and Werner Liniger.
\newblock Exact analysis of an interaction {Bose} gas. {I}: {The} general
  solution and the ground state.
\newblock {\em Phys. Rev., II. Ser.}, 130:1605--1616, 1963.

\bibitem{LieSeiYng-05}
Elliott~H. Lieb, Robert Seiringer, and Jakob Yngvason.
\newblock {Justification of $c$-Number Substitutions in Bosonic Hamiltonians}.
\newblock {\em Phys. Rev. Lett.}, 94:080401, Mar 2005.

\bibitem{LY}
Elliott~H. Lieb and Jakob Yngvason.
\newblock Ground {S}tate {E}nergy of the {L}ow {D}ensity {B}ose {G}as.
\newblock {\em Phys. Rev. Lett.}, 80(12):2504--2507, mar 1998.

\bibitem{zbMATH01638213}
Elliott~H. Lieb and Jakob Yngvason.
\newblock The ground state energy of a dilute two-dimensional {Bose} gas.
\newblock {\em J. Stat. Phys.}, 103(3-4):509--526, 2001.

\bibitem{zbMATH07277873}
Simon Mayer and Robert Seiringer.
\newblock The free energy of the two-dimensional dilute {Bose} gas. {II}:
  {Upper} bound.
\newblock {\em J. Math. Phys.}, 61(6):061901, 17, 2020.

\bibitem{Ruelle}
David Ruelle.
\newblock Statistical mechanics: {R}igorous results, 1969.

\bibitem{zbMATH05279032}
Robert Seiringer.
\newblock Free energy of a dilute {Bose} gas: lower bound.
\newblock {\em Commun. Math. Phys.}, 279(3):595--636, 2008.

\bibitem{YY}
Horng-Tzer Yau and Jun Yin.
\newblock The {S}econd {O}rder {U}pper {B}ound for the {G}round {E}nergy of a
  {B}ose {G}as.
\newblock {\em Journal of Statistical Physics}, 136(3):453--503, jul 2009.

\bibitem{yin}
Jun Yin.
\newblock Free energies of dilute {Bose} gases: upper bound.
\newblock {\em J. Stat. Phys.}, 141:683--726, 2010.

\end{thebibliography}

\end{document}